\documentclass[11pt]{article}
\usepackage{sectsty}

\usepackage{graphicx}

\usepackage{amsmath,amssymb}
\usepackage{bm}
\usepackage{enumerate}
\usepackage{amsthm}
\usepackage[all,2cell]{xy}
\usepackage{mathrsfs}
\usepackage{mathtools}
\usepackage{tikz-cd}
\usepackage{tikz}
\usetikzlibrary{shapes.multipart}
\usetikzlibrary{patterns,patterns.meta}
\usepackage{xcolor}
\definecolor{green}{rgb}{0.0, 0.5, 0.0}
\DeclareUnicodeCharacter{0301}{*************************************}
\usepackage[normalem]{ulem}
\usepackage{float}

\usepackage[colorlinks=true]{hyperref}
\hypersetup{urlcolor=blue, citecolor=red, linkcolor=blue}

\usetikzlibrary{cd}

\usepackage[a4paper, left=25mm, right=25mm, top=30mm, bottom=30mm]{geometry}

\usepackage{marginnote}

\usepackage{imakeidx}
\makeindex[intoc]

\usepackage[nottoc]{tocbibind}
\usepackage[colorinlistoftodos]{todonotes} 

\newtheorem{thm}{Theorem}[section]
\newtheorem{dfn}[thm]{Definition}
\newtheorem{prop}[thm]{Proposition}
\newtheorem{cor}[thm]{Corollary}
\newtheorem{lem}[thm]{Lemma}

\newcommand\restr[2]{{
  \left.\kern-\nulldelimiterspace 
  #1 
  \right|_{#2} 
}}

\renewcommand{\d}{\mathrm{d}}

\newcommand{\inn}{i}
\newcommand{\Lie}{\mathscr{L}}

\DeclareMathAlphabet{\mathpzc}{OT1}{pzc}{m}{it}
\def\d{\mathrm{d}}

\mathtoolsset{showonlyrefs}

\let\oldemph\emph
\let\emph\textbf

\title{{\sffamily From Variational Principles to Geometry}}

\author{{\sffamily 
Jordi Gaset Rifà%
\thanks{e-mail:
   jordi.gaset@cunef.edu \ ORCID: 0000-0001-8796-3149}\ ,\
}
\\[1ex]
\normalsize\itshape\sffamily 
Department of Mathematics,
CUNEF Univ.,
Madrid, Spain.
}

\date{}

\begin{document}

\maketitle

\begin{abstract}
    A method to construct a geometric structure with the same solutions as a given variational principle is presented. The method applies to large families of variational principles. In particular, the known results that assign cosymplectic geometry to Hamilton's principle and cocontact geometry to Herglotz's principle for regular Lagrangians are recovered. The unified Lagrangian-Hamiltonian formalism is also recovered via the absorption of the holonomy conditions. The method is applied to singular time-dependent Lagrangians, proving that they can always be described with a (pre)cosymplectic structure, although it is not always given by the Lagrangian $2$-form. When applied to singular action-dependent Lagrangians, the method does not always lead to (pre)cocontact geometry. In these cases, the resulting geometry associated with the Herglotz's variational principle is new.
\end{abstract}




{\setcounter{tocdepth}{2}
\def\baselinestretch{1}
\small
\def\addvspace#1{\vskip 1pt}
\parskip 0pt plus 0.1mm
\tableofcontents
}

\newpage

\section{Introduction}

Geometric mechanics studies dynamical systems with geometric objects. The discipline stems from analytic mechanics and the study of physical problems. The prototypical theory in geometric mechanics is symplectic geometry. It is the modern incarnation of classical mechanics. It describes in a unified geometric language Newtonian mechanics, Hamilton's variational principle, Lagrangian and Hamiltonian mechanics and Noether's theorem, among other topics. 

The geometric description provides an intrinsic formulation and leads to more powerful results. For instance, the Meyer-Marsden-Weinstein reduction theorem \cite{MW} shows how to reduce a system to a lower-dimensional one using its symmetries. If the system is integrable, the Liouville-Mineur-Arnold theorem \cite{AM-78} states that one can find canonical coordinates adapted to the dynamics. In some instances, a complete classification of integrable dynamical systems is possible through the use of the momentum map and polytopes \cite{Atiyah}.

However, not all dynamical systems can be described with a symplectic structure. To overcome this
problem, new geometric structures have been developed as needed. When the Lagrangian is not regular, the corresponding structure is presymplectic geometry. Another natural generalization of symplectic geometry is cosymplectic geometry, which is the natural framework when the system is explicitly time-dependent. See \cite{munoz-lecanda_geometry_2024} for a recent book presenting these structures. When the symplectic manifold has a distinguished hypersurface, due to having a boundary or singularities, it is natural to use $b$-symplectic geometry \cite{GMPS-2015}. For field theories, several extensions of symplectic geometry exist, like $k$-symplectic and multisymplectic (see \cite{LeSaVi2016, Roman-Roy:2005vwe} and references therein, respectively).

Symplectic geometry is inherently adapted to describe
conservative systems, therefore, open systems are best described with other geometries. In particular, contact geometry has recently received a lot of attention \cite{Bravetti2017, LL-2018, GGMRR-2019b}. It has been extended to precontact geometry \cite{LL-2019} for singular action-dependent Lagrangians and cocontact geometry \cite{LGGMR-2022} for time and action-dependent Lagrangians. Contact analogues of $b$-symplectic \cite{MO-2023}, $k$-symplectic \cite{GGMRR-2019,GGMRR-2020} and multisymplectic \cite{LGMRR-2022b,deLeon:2024ztn, SurveyAdFT, bracketsmultico} geometries have been recently proposed. 

Powerful geometric techniques, like the Meyer-Marsden-Weinstein theorem, the Liouville-Mineur-Arnold theorem or the momentum map explained before, have to be adapted for each geometric structure. For instance, the Liouville-Mineur-Arnold theorem for contact systems needs a non-trivial modification \cite{CLLL}. This situation compartmentalizes the field of geometric mechanics, as each geometric structure is developed independently.

The existence of different compartments inside geometric mechanics poses new questions. How are the different structures related? Are they disjoint? A result in this line of thought is the Jacobi structure introduced by Lichnerowicz \cite{Lichnerowicz}, which encompasses several structures, including symplectic, contact and Poisson. More recently, there have been studies on the relation between (pre)cosymplectic and (pre)contact geometries \cite{deLeon:2016vms,grabowska}.

The geometric structure associated with a dynamical system is especially relevant if it faithfully reflects its properties; hence, we must be able to identify the correct geometric structure for it. Moreover, the particular geometry identified should give the principal properties of the system. In other words, the assignment of geometric structures should provide a broad classification of dynamical systems. 

It could happen that different parts of a dynamical system behave in qualitatively different ways. For instance, when some technical conditions hold, the dynamics of a contact system can be split into Reeb dynamics and Liouville dynamics \cite{BLMP-2020}. There also exist the processes of contactification and symplectification, where one can embed a contact system into a symplectic system, and viceversa, by adding or eliminating a variable \cite{GGKM-2024,Sloan:2024kzb} (see \cite{bracketsmultico, Bell:2025kxx} for a field theoretical analogue). These works are a small selection of the ongoing debate on the relevance of contact dynamics with respect to symplectic dynamics.

Another relevant question is how to assign a geometry to a new dynamical system. So far, this process is done case-by-case; selecting a geometric structure and then attempting to solve the corresponding inverse problem \cite{C1981,LGL-2021}. Even when a Lagrangian and a variational principle are known, the geometry can not always be recovered. For instance, the Lagrangians
\begin{align}\label{eq:lagraros}
    L(t,q,v)=tv\,\quad\text{ and }\quad L(q,v,s)=sv 
\end{align}
fail to conform to the standard precosymplectic and precontact geometries, respectively, against expectation.

In this work, we will give a method to assign a geometry to a dynamical system that has a variational principle. We do not prove that the associated geometry is unique. Nevertheless, for the systems where there is an agreed-upon geometry, the method gives the same result. As an application, we will derive a geometric structure for the Lagrangians \eqref{eq:lagraros}.

The paper is structured as follows. In section \ref{sec:2} we present variational principles with constraints, and show how to rewrite them in geometric terms. In sections \ref{sec:trans}, \ref{sec:nonhol}, \ref{sec:vak} and \ref{sec:nhvak} we present the first step of the method in different situations depending on the constraints. The last step is presented in section \ref{sec:abs}, where we show how to absorb the constraints. Finally, in sections \ref{sec:singtime} and \ref{sec:singaction} we apply the method to uncover the geometry of Lagrangians of the form \eqref{eq:lagraros}.

\subsection*{Notation}
Consider the fiber bundle $\pi:E\rightarrow \mathbb{R}$. Let $t$ be the coordinate of $\mathbb{R}$.
\begin{itemize} 
\item $\hat\tau=\d t$ and $\tau=\pi^*\hat\tau$.
\item $I_\gamma$: domain of the section $\gamma$. Along the text we will assume they are open intervals of $\mathbb{R}$.
    \item $\Gamma_U(\pi)$: local sections of $\pi$ whose domain contains $U\subset \mathbb{R}$.
    
    \item $V(\pi)$: vertical distribution of $\pi$.
    \begin{align}
        V(\pi)=\bigcup_{x\in E}V_x(\pi)=\bigcup_{x\in E}\{v\in T_x E\,|\,\d\pi(v)=0\}\,.
    \end{align}
    \item  $B(\pi)$: $\pi$-basic functions of $E$.
        \begin{align}
    B(\pi)=\{g\in C^\infty(E)| 
    \,\forall x\in E \text{ and }\,\forall v\in V_x(\pi)\,,\,\inn_v\d g=0\}\,.       
\end{align}
    
    \item $\Gamma(D)$: sections (vector fields) of the distribution $D\subset TE$.    
    \item$\langle \alpha_i \rangle_{i=1,\cdots,k}$: all linear combinations of the elements $\alpha_i$ with coefficients in $C^\infty(E)$.
    
    \item$\langle S \rangle^{B(\pi)}$: all finite linear combinations of the elements of $S$ with coefficients in $B(\pi)$.

\end{itemize}

\section{From a Variational Formulation to a Geometric Formulation}\label{sec:2}

In this section we will see how to write variational problems in a geometric language and state the aim of the method presented in this work.

The curves $\gamma$ will be understood as local sections of fiber bundles over $\mathbb{R}$, $\pi:E\rightarrow \mathbb{R}$. They are trivial and we have a global coordinate on $\mathbb{R}$ which we will denote $t$. The corresponding volume form will be denoted $\hat{\tau}=\d t$, and $\tau=\pi^*\hat{\tau}$. Moreover, there exists a (non-unique) global vector field $R_t\in \mathfrak{X}(E)$ such that $d\pi(R_t)=\frac{\d}{\d t}$. In particular, $\inn_{R_t}\tau=1$. The domain of the curves will be denoted $I_\gamma\subset\mathbb{R}$, and we will assume it is an open interval.

We will use the fiber bundle approach in this work because it conforms better with geometric mechanics literature, it is well-suited to time-dependent systems and it is more amenable to future generalizations.

\subsection{Variational Problems}

Let $\mathcal{K}$ be a compact subset of $\mathbb{R}$. Given a $1$-form $\theta$ on $E$, we have the action $\mathcal A(\mathcal{K},\cdot):\Gamma_\mathcal{K}(\pi)\rightarrow \mathbb{R}$ given by
$$
\mathcal A(\mathcal{K},\gamma)=\int_\mathcal{K} \gamma^*\theta\,.
$$

A variation of a field $\gamma$ will be given by $\gamma_s=\psi_s\circ\gamma$, where $\psi_s$ is the element of the one-parameter group of diffeomorphisms of $E$ defined by the smooth vector field $\xi$ corresponding to parameter value $s$. The variation of $\mathcal{A}$ on $\gamma$ in the direction of $\xi$ over $\mathcal{K}$ is (see, for instance, \cite{Krupka,munoz-lecanda_geometry_2024}) 
\begin{align}
    \left.\frac{\d}{\d s}\right|_{s=0}\mathcal A(\mathcal{K},\gamma_s)=\int_\mathcal{K} \gamma^*\Lie_{\xi}\theta\,.
\end{align}

The direction of the variation $\xi$, also called just a variation, is,  more generally, a vector field along a curve $\gamma$ which vanishes at the boundary of $\mathcal{K}$  \cite{gracia_geometric_2003}. In this work, it will be enough to consider the variations as vector fields on $\mathfrak{X}(E)$. It is important to note that, depending on the variational problem, the variations represent only a subset of vector fields $\textup{Adm}\subset\mathfrak{X}(E)$. The elements of Adm are called the \textbf{admissible variations}. Given a curve $\gamma$ with domain $I_\gamma$ and a compact set $\mathcal{K}\subset I_\gamma$, we define
\begin{align}
    \textup{Adm}(\gamma,\mathcal{K}):=\{\xi\in\textup{Adm} \,|\, \xi(\gamma(\partial\mathcal{K}))=0\}
\end{align}
A similar object is defined in \cite{Krupka}. As a consequence that the variations vanish at the boundary, in this work we will not concern ourselves with boundary terms. This will make the exposition clearer. Boundary conditions are included in a variational problem in diverse ways \cite{TTW-2000}. They are relevant in applications \cite{DH-2009} and can be implemented geometrically \cite{Margalef-Bentabol:2020teu}.



Last but not least, it is common to impose extra conditions on the sections, called constraints. The constraints considered in this work are given by functions or $1$-forms on $E$, which we will denote $J\subset\Omega^\bullet(E)$. The constraints play another crucial role, as commonly the admissible variations are defined as those vector fields compatible with the constraints.

After these considerations, in this work we will consider the following notion of variational principle:

\begin{dfn}A \textbf{variational problem or principle} (VP) on a fiber bundle $\pi:E\rightarrow M$ is a triple $(\theta,J,\textup{Adm})$, where $\theta\in\Omega^1(E)$, $J\subset \Omega^\bullet(E)$ and $\textup{Adm}\subset\mathfrak{X}(E)$.
\end{dfn}

\begin{dfn}\label{dfn:solVP}
A section $\gamma$ of $\pi$ with domain $I_\gamma\subset\mathbb{R}$ is a \textbf{solution} of the variational problem $(\theta,J,\textup{Adm})$ if
\begin{enumerate}
    \item $\gamma^*\alpha=0$ for all $\alpha\in J$.
    \item For all compact sets $\mathcal{K}\subset I_\gamma$,
$$
\int_\mathcal{K} \gamma^*\Lie_\xi\theta=0\,, \quad \forall\xi\in \textup{Adm}(\gamma,\mathcal{K})\,. 
$$
\end{enumerate}
\end{dfn}

The set of all solutions will be denoted $\text{Sol}(\theta,J,\textup{Adm})$.

\subsection{Geometric Formulation of Variational Problems}

In the geometric study of dynamical systems, a different formulation closer to the spirit of differential geometry is used. It hinges on the fundamental lemma of the calculus of variations.

The set of $\pi$-basic functions will be denoted:
\begin{align}
    B(\pi)=\{g\in C^\infty(E)| 
    \,\forall x\in E \text{ and }\,\forall v\in V_x(\pi)\,,\,\inn_v\d g=0\}\,.
\end{align}

\begin{thm}\label{thm:base} Consider the VP $(\theta,J,\textup{Adm})$. If Adm is $B(\pi)$-linear, then a local section $\gamma$ of $\pi$ is a solution of $(\theta,J,\textup{Adm})$ if and only if
\begin{enumerate}
    \item $\gamma^*\alpha=0$, for all $\alpha\in J$.
    \item $\gamma^*\inn_{\xi}\d\theta=0$, for all $\xi\in\textup{Adm}$.
\end{enumerate}  
\end{thm}
\begin{proof}
The first condition is the same as the first condition of \ref{dfn:solVP}, so we will concentrate on the second one.
Let $\gamma$ be a solution of the variational problem, with domain $I_\gamma$. Let $\hat{g}\in C^\infty(\mathbb{R})$ be a compactly supported function with $Supp(\hat{g})\subset\mathcal{K}\subset I_\gamma$, for some compact set $\mathcal{K}$ inside $I_\gamma$. Define $g=\hat{g}\circ\pi\in C^\infty(E)$. It is straightforward from its definition that $g\in B(\pi)$ and $g\circ\gamma=\hat{g}$. For any $\xi\in\textup{Adm}$, the product $g\xi$ is also an admissible variation because Adm is $B(\pi)$-lineal. Moreover, $g\xi\in\textup{Adm}(\gamma,\mathcal{K})$ because $\hat{g}$ has support inside $\mathcal{K}$. The corresponding variation on the solution $\gamma$ over $\mathcal{K}$ is:
\begin{align}
  0&=\int_\mathcal{K} \gamma^*\Lie_{g\xi}\theta
  =\int_\mathcal{K} \gamma^*\left(\inn_{g\xi}\d\theta+\d\inn_{g\xi}\theta\right)
=\int_\mathcal{K} \gamma^*\inn_{g\xi}\d\theta+\int_{\partial \mathcal{K}} \gamma^*\inn_{g\xi}\theta
\\
&=\int_\mathcal{K}\hat{g} \gamma^*\inn_{\xi}\d\theta+\int_{\partial \mathcal{K}} \hat{g}\gamma^*\inn_{\xi}\theta
=\int_\mathcal{K}\hat{g} \gamma^*\inn_{\xi}\d\theta
=\int_{I_\gamma}\hat{g} \gamma^*\inn_{\xi}\d\theta
\end{align}
Since this holds for all compactly supported smooth functions $\hat{g}$ on $I_\gamma$, according to the fundamental lemma of the calculus of variations, it must be that
\begin{align}\label{eq:segonacondicio}
    \gamma^*\inn_\xi\d\theta=0\,,\,\text{ for all } \xi\in\textup{Adm}.
\end{align}

Conversely, let $\gamma$ be a curve that satisfies \eqref{eq:segonacondicio}. For any compact set $\mathcal{K}\subset I_\gamma$ and any vector field $\xi\in \textup{Adm}(\gamma,\mathcal{K})$:
\begin{align}
  \int_\mathcal{K} \gamma^*\Lie_{\xi}\theta
  =\int_\mathcal{K} \gamma^*\left(\inn_{\xi}\d\theta+\d\inn_{\xi}\theta\right)
=\int_\mathcal{K} \gamma^*\inn_{\xi}\d\theta+\int_{\partial \mathcal{K}} \gamma^*\inn_{\xi}\theta=0\,.
\end{align}

\end{proof}

This result motivates the following definitions.

\begin{dfn}A \textbf{geometric variational problem or principle} (GVP) on a fiber bundle $\pi:E\rightarrow \mathbb{R}$ is given by the triple $(\Omega,J,\textup{Adm})$, where $\Omega\in\Omega^2(E)$, $J\subset \Omega^\bullet(E)$ and $\textup{Adm}\subset\mathfrak{X}(E)$.
\end{dfn}

Notice that we are not requiring that $\Omega$ is a closed form.

\begin{dfn}\label{dfn:solsec}
A local section $\gamma$ of $\pi$ is a \textbf{solution} of the geometric variational problem $(\Omega,J,\textup{Adm})$ if
\begin{enumerate}
    \item $\gamma^*\alpha=0$ for all $\alpha\in J$.
    \item $\gamma^*\inn_{\xi}\Omega=0$, for all $\xi\in\textup{Adm}$.
\end{enumerate}
\end{dfn}

The set of all solutions will be denoted $\text{Sol}(\Omega,J,\textup{Adm})$. In the geometric approach it is convenient to think about solutions as vector fields. See, for instance, the Fourth Postulate of Hamiltonian mechanics in \cite{munoz-lecanda_geometry_2024}. Recall that, if $\gamma$ is an integral section of $Z$, then the time-translated curve $\gamma_a(t):=\gamma(t+a)$ is also an integral section of $Z$ for all $a\in\mathbb{R}$.

\begin{dfn}\label{dfn:solvector}
A vector field $Z\in \mathfrak{X}(E)$ is a \textbf{solution} of the geometric variational problem $(\Omega,J,\textup{Adm})$ if its integral curves are, up to time-translation, solutions of $(\Omega,J,\textup{Adm})$. 
\end{dfn}

\begin{lem}
A vector field $Z\in \mathfrak{X}(E)$ is a \textbf{solution} of the geometric variational problem $(\Omega,J,\textup{Adm})$ if and only if
\begin{enumerate}
    \item $\inn_Z\alpha=0$ for all $\alpha\in J$.
    \item $\inn_Z\inn_{\xi}\Omega=0$, for all $\xi\in\textup{Adm}$.
    \item $\inn_Z\tau=1$.
\end{enumerate}
\end{lem}
\begin{proof}
   Let $\gamma$ be an integral curve of $Z$, and $t\in I_\gamma$. We can write any vector $v\in T_t\mathbb{R}$  as $v=\rho \left.\frac{\d}{\d t}\right|_t$ for a certain real number $\rho$. Given a $1$-form $\varepsilon\in\Omega^1(E)$ 
   \begin{align}\label{eq:camposproof}
       \gamma^*\varepsilon_t(v)=\varepsilon_{\gamma(t)}(\d \gamma_t(v))=\varepsilon_{\gamma(t)}(\rho Z_{\gamma(t)})=\rho(\inn_Z\varepsilon)_{\gamma(t)}\,.
   \end{align}
   Since $Z$ is smooth, it has an integral curve passing through any point of $E$. This proves that the first and second equations are equivalent to the first and second equations of \ref{dfn:solsec}.

   Moreover, from \eqref{eq:camposproof} with $\varepsilon= \tau$ we deduce that
   \begin{align*}
       \inn_Z\tau=1\Longleftrightarrow \gamma^*\tau=\hat\tau\,. 
   \end{align*}
   If $\gamma^*\tau=\hat\tau$, then, for all $t\in I_\gamma$, $\gamma^*t=t-a$ for some constant $a\in\mathbb{R}$ independent of $t$. Then the time-translated curve $\gamma_a$ is a section of $\pi$.
\end{proof}

The equations labelled as 2 in \ref{dfn:solsec} and in \ref{dfn:solvector} are called the \textbf{equations of motion}.

\subsection{Relations between Variational Problems}

Let $Sol_1$ and $Sol_2$ be the sets of solutions of two variational problems $V_1$ and $V_2$ on the bundles $\pi_1:E_1\rightarrow\mathbb{R}$ and $\pi_2:E_2\rightarrow\mathbb{R}$ respectively. The problems $V_1$ and $V_2$ can be VP or GVP. An smooth bundle map $\mu:E_1\rightarrow E_2$ acts on local sections over a set $U\subset\mathbb{R}$  by composition:
\begin{align}
    \tilde{\mu}:\Gamma_U(\pi_1)&\rightarrow \Gamma_U(\pi_2)
    \\
    \gamma&\mapsto \tilde{\mu}(\gamma)=\mu\circ\gamma
\end{align}

and $I_{\tilde{\mu}(\gamma)}=I_\gamma$.

\begin{dfn}\label{dfn:equal} Consider two (geometric) variational problems $V_1$, $V_2$:
\begin{itemize}
    \item $V_1$ is \textbf{dynamically equivalent} to $V_2$ if $Sol_1=Sol_2$.
    \item $V_1$ is a \textbf{subsystem} of $V_2$ if $Sol_1\subset Sol_2$.  In this situation, we also say that $V_2$ is a \textbf{supersystem} of $V_1$.
    \item $V_1$ is \textbf{dynamically $\mu$-equivalent} to $V_2$ if there exists a smooth bundle map $\mu:E_1\rightarrow E_2$ such that $\tilde{\mu}(Sol_1)=Sol_2$ and $\tilde{\mu}$ is injective on $Sol_1$. If, moreover, $\mu$ is a diffeomorphism, then they are called \textbf{diffeomorphic}; in particular, $E_1\approx E_2$. 
    \item $V_1$ is an \textbf{extension} of $V_2$ if there exists a smooth map $\mu:E_1\rightarrow E_2$ such that $\tilde{\mu}(Sol_1)= Sol_2$ although $\tilde\mu$ may not be injective on $Sol_1$. In this situation, we also say that $V_2$ is a \textbf{reduction} of $V_1$.
    \item $V_1$ is a \textbf{restriction} of $V_2$ if there exists a smooth bundle map $\mu:E_1\rightarrow E_2$ such that $\tilde{\mu}(Sol_1)\subset Sol_2$ and $\tilde{\mu}$  is injective on $Sol_1$. 

\end{itemize}
\end{dfn}

The proposed definition of dynamically equivalent systems is closely related to other notions of ``equivalence'' present in the literature. For instance,  equivalent Lagrangians and Lepage equivalence \cite{Krupka} refer to different formulations that lead to the same differential equations (see also, for instance, \cite{GBVP,LGL-2021}). Moreover, the adjective ``equivalent'' is commonly used in the literature in its non-technical meaning. Therefore, we decided to use ``dynamically equivalent'' to avoid confusion. In the same fashion, we propose the adjective ``dynamically $\mu$-equivalent''. In \cite{LGL-2021} ``dynamically $\mu$-equivalent'' is denoted as ``dynamically similar'', but this name is ambiguous \cite{Bell:2025kxx}. Notice that, in general, even if a transformation relates the solutions one-to-one, both systems can have qualitatively different dynamics unless the transformation also conserves relevant structures (see, for instance, natural transformation in Lagrangian mechanics \cite{munoz-lecanda_geometry_2024}).

Now we have the language to state that more variations imply more equations and less solutions.
\begin{lem}\label{lem:subvar} Consider the GVPs $(\Omega,J,\textup{Adm}_1)$ and $(\Omega,J,\textup{Adm}_2)$ over the same bundle. If $\textup{Adm}_1\subset\textup{Adm}_2$, then $(\Omega,J,\textup{Adm}_2)$ is a subsystem of $(\Omega,J,\textup{Adm}_1)$\,.
\end{lem}
\begin{proof}
The result follows from the observation that both GVPs have the same constraints and any equation of motion of $(\Omega,J,\textup{Adm}_1)$ is also an equation of motion of $(\Omega,J,\textup{Adm}_2)$. 
\end{proof}

Notice that this result also holds for VPs.

\begin{lem}\label{lem:subvar} The GVP $(\Omega,J,\textup{Adm})$ is dynamically equivalent to the GVP $(\Omega,J,\langle\textup{Adm}\rangle)$.
\end{lem}
\begin{proof}
Since $\textup{Adm}\subset \langle\textup{Adm}\rangle$, by the previous lemma $Sol(\Omega,J,\langle\textup{Adm}\rangle)\subset Sol(\Omega,J,\textup{Adm})$. Let $\gamma$ be a solution of $(\Omega,J,\textup{Adm})$ and $\xi\in \langle\textup{Adm}\rangle$. By construction, there exists a function $f\in C^\infty(E)$ and a vector field $\xi_0\in\textup{Adm}$ such that $\xi=f\xi_0$. Then,
\begin{align}
    \gamma^*\inn_\xi\Omega=\gamma^*\inn_{f\xi_0}\Omega=\gamma^*(f)\gamma^*\inn_{\xi_0}\Omega=0\,.
\end{align}
Therefore, $Sol(\Omega,J,\langle\textup{Adm}\rangle) =Sol(\Omega,J,\textup{Adm})$.
\end{proof}

\subsection{The Geometric Mechanics Approach}

The GVPs frames the variational problems in a geometric language, but the underlying structure is not apparent because the dynamical information is encoded in three different objects: the $2$-form $\Omega$, the constraints $J$ and the admissible variations Adm. We want to find a new GVP that is dynamically equivalent to the original one and such that all the dynamical information is contained only in one of the objects. 

An example of this idea is the Exterior Differential Systems approach \cite{Griffiths}. It is based on the following observation. The equation of motion in \ref{dfn:solsec} can be thought of as imposing the constraint $\inn_\xi\Omega$. Therefore, a GVP $(\Omega,J,\textup{Adm})$ is dynamically equivalent to the GVP $(0,J\cup\{\inn_\xi\Omega\}_{\xi\in\textup{Adm}},\emptyset)$. Consequently, in the Exterior Differential Systems approach, the central objects of study are differential ideals.

On the other hand, in symplectic systems the object of interest is the symplectic form. This leads to a different approach to study the dynamics of a variational problem which we call the \textbf{geometric mechanics approach}. There are two levels to this approach. First, to find a $2$-form ${\Omega}_1\in\Omega^2(E)$ such that
\begin{align}
 (\Omega,J,\textup{Adm})\text{ is dynamically equivalent to }({\Omega}_1, J, \mathfrak{X}(E))\,.   
\end{align}

Secondly, to find a $2$-form ${\Omega}_2\in\Omega^2(E)$ such that
\begin{align}
 (\Omega,J,\textup{Adm})\text{ is dynamically equivalent to }({\Omega}_2, \emptyset, \mathfrak{X}(E))\,.   
\end{align}

In this way, the dynamics of the variational problem is completely characterized by the $2$-form ${\Omega}_2$, or the pair $(\Omega_1, J)$. This is the process of constructing a geometric structure for a variational problem.

In this work we will show how to construct ${\Omega}_1$ in four relevant situations, depending on how the constraints select the admissible variations: transversality constraints (Section \ref{sec:trans}), nonholonomic constraints (Section \ref{sec:nonhol}), vakonomic constraints (Section \ref{sec:vak}) and a mix of all of them (Section \ref{sec:nhvak}). A detailed presentation of some of these constraints and their relation to variational principles can be found here \cite{gracia_geometric_2003}. We will defined them in each section, as the definitions used in this work are modifications of the ones present in \cite{gracia_geometric_2003} to adapt them to our setting.

In section \ref{sec:abs} we show how to construct $\Omega_2$, at the cost of adding new variables and only obtaining an extension, not a dynamical equivalence (in the sense of definition \ref{dfn:equal}).

\section{Transversality Constraint}\label{sec:trans}

Consider the fiber bundle $\pi:E\rightarrow \mathbb{R}$. Remember that $\hat\tau$ is a volume form of $\mathbb{R}$ and $\tau=\pi^*\hat{\tau}\in \Omega^1(E)$. The solutions of a GVP or VP are sections of $\pi$, which imposes a ``constraint'' on $\gamma$. Namely, $\pi\circ\gamma=Id_{I_\gamma}$.

Commonly in a variational principle, the variation of a section is still a section \cite{gracia_geometric_2003,munoz-lecanda_geometry_2024}. In our formulation this is encoded as demanding that the admissible variations are vertical vector fields: Adm$\subset\Gamma(V(\pi))$. Notwithstanding, it is common in symplectic geometry to consider all vector fields. Notice that symplectic geometry describe autonomous systems: the importance of taking only vertical vector fields  becomes relevant for non-autonomous systems.

The objective of this section is, given a GVP $(\Omega,J,\Gamma(V(\pi)))$, to find a $2$-form $\bar{\Omega}$ such that the GVP $(\bar{\Omega},J,\mathfrak{X}(E))$ is dynamically equivalent. It turns out that, imposing a small number of extra conditions, we can take $\bar{\Omega}=\Omega$.

Take a vector field $R_t\in\mathfrak{X}(E)$ such that $\inn_{R_t}\tau=1$ everywhere, and consider the decomposition $\Omega= \omega-\sigma_t\wedge\tau$, where 
\begin{align}\label{eq:formstrans}
    \sigma_t=\inn_{R_t}\Omega\,,\quad\text{ and }\quad\omega=\Omega+\sigma_t\wedge\tau\,.
\end{align}
By construction, $\inn_{R_t}\sigma_t=0$ and $\inn_{R_t}\omega=0$. 

\begin{thm} \label{theo:trans}The GVP  $(\Omega,J,\Gamma(V(\pi)))$ is dynamically equivalent to the GVP $(\Omega,J,\mathfrak{X}(E))$ and they have a vector field solution, if and only if there exists $X\in \Gamma(V(\pi))$ such that $\inn_X\omega=\sigma_t$ and $\inn_{R_t-X}\alpha=0$ for all $\alpha\in J$.
\end{thm}
\begin{proof}  Since $\Gamma(V(\pi))\subset\mathfrak{X}(E)$, using lemma \ref{lem:subvar} we have that 
    \begin{align}
        \text{Sol}(\Omega, J, \Gamma(V(\pi)))\supset\text{Sol}(\Omega, J, \mathfrak{X}(E))\,.    
    \end{align}
Assume there exists $X\in \Gamma(V(\pi))$ such that $\inn_X\omega=\sigma_t$. Contracting with $X$ we deduce that $\inn_X\sigma_t=0$. Then,
\begin{align}
 \inn_X\omega=\sigma_t\Rightarrow\inn_X\omega-\inn_X\sigma_t\tau=\sigma_t\Rightarrow \inn_X\Omega=\inn_{R_t}\Omega\,.    
\end{align}
Let $\gamma$ be a solution of $(\Omega,J,\Gamma(V(\pi)))$, then
\begin{align}
\gamma^*\inn_{R_t}\Omega=\gamma^*\inn_X\Omega=0\,.
\end{align}
Since any vector field $Y$ can be decomposed as $Y=S+\inn_Y\tau R_t$ for some vertical vector field $S$, we deduce that $\gamma$ is a solution of $(\Omega,J,\mathfrak{X}(E))$. Therefore, $(\Omega,J,\mathfrak{X}(E))$ is dynamically equivalent to $(\Omega,J,\Gamma(V(\pi)))$. Moreover, $\inn_{R_t-X}\tau=1$, $\inn_{R_t-X}\alpha=0$ for all $\alpha\in I$ and $\inn_{R_t-X}\Omega=0$, and thus both variational problems have a vector field solution.

Conversely, let $Z$ be a vector field solution. It can be decomposed as $Z=R_t-X$, with $X=R_t-Z$ vertical. Then, since it is a solution of $(\Omega,J,\mathfrak{X}(E))$
$$
\inn_Z\Omega=0\Rightarrow -\inn_X\omega+\inn_X\sigma_t\tau+\sigma_t=0\,.
$$
Contracting with $R_t$ we derive $\inn_X\sigma_t=0$. It follows that $ \inn_X\omega=\sigma_t$.
\end{proof}

With the hypotheses of theorem \ref{theo:trans}, the equations of motion for a vector field $Z$ can be equivalently written as:
$$
\inn_Z\Omega=0\Leftrightarrow\inn_Z\omega=-\sigma_t.
$$
In particular, we have that 
\begin{align}\label{eq:conservationLaw}
\inn_Z\sigma_t=0\,.    
\end{align}
For symplectic systems, this is the energy conservation law and, for contact systems, it is the energy dissipation law.

Since $\pi:E\rightarrow\mathbb{R}$ is trivial, we have that $E=\mathbb{R}\times N$, with $N$ a smooth manifold. The projection to $N$ will be denoted $\pi_N:E\rightarrow N$. Define $\omega$ and $\sigma_t$ as in \eqref{eq:formstrans} with $R_t=\frac{\partial}{\partial t}$. We say that a GVP $(\Omega,\emptyset,\Gamma(V(\pi)))$ is  \textbf{autonomous} if $\Lie_{R_t}\omega=0$ and $ \Lie_{R_t}\sigma_t=0$. Then, we can project $\omega$ and $\sigma_t$ to forms $\hat{\omega}$ and $\hat{\sigma_t}$ on $N$ 

\begin{lem} Let $(\Omega,\emptyset,\Gamma(V(\pi)))$ be an autonomous variational problem over a trivial bundle $E=\mathbb{R}\times N$. Assume that $\hat{\omega}$ is a non-degenerate $2$-form of $N$. Then, the GVP  $(\Omega,\emptyset,\Gamma(V(\pi)))$ has a unique vector field solution with the form $Z=R_t-X$, where $X$ is $\pi_N$ related to a vector field $\hat{X}\in\mathfrak{X}(N)$ that satisfies the equation:
$$
\inn_{\hat{X}}\hat{\omega}=\hat{\sigma_t}.
$$
\end{lem}

\begin{proof}
    Since $\hat{\omega}$ is non-degenerate, there exists $\hat{X}$ such that $\inn_{\hat{X}}\hat{\omega}=\hat{\sigma_t}$. In particular, $\inn_{\hat{X}}\hat{\sigma_t}=0$. Let $X$ be the lift of $\hat X$ to $\mathbb{R}\times N$ with the $0$-section of $\pi_N$. Then, $Z=R_t-X$ is a vector field solution. Indeed,
    \begin{align*}
        \inn_Z\Omega=-\inn_X\omega-\sigma_t+\inn_X\sigma_t=-\pi_N^*(\inn_{\hat{X}}\hat{\omega}+\hat{\sigma_t}-\inn_{\hat{X}}\hat{\sigma_t})=0\,.
    \end{align*}
    In order to prove uniqueness, assume there exists another $Z_2\in\mathfrak{X}(E)$ such that $\inn_{Z_2}\tau=1$ and $\inn_{Z_2}\Omega=0$. Then, $\inn_{Z-Z_2}\tau=0$ and $\inn_{Z-Z_2}\omega=0$. Since $\hat{\omega}$ is non-degenerate, $\ker(\omega)=V(\pi_N)$. Therefore, $Z-Z_2=0$.
\end{proof}

\paragraph{Example: Hamiltonian symplectic and co-symplectic systems.} Consider the bundle $\pi:E=\mathbb{R}\times T^*Q\rightarrow \mathbb{R}$. Let $\hat{\theta}$ be the canonical $1$-form of $T^*Q$ and denote by $\theta$ its pullback to $E$. Consider a function $H:E\rightarrow \mathbb{R}$. Let $\mathcal{K}$ be a compact subset of $\mathbb{R}$. The variational problem for Hamiltonian systems is given by the action:
$$
\mathcal{A}(\mathcal{K},\gamma)=\int_\mathcal{K} \gamma^*(\theta-H\tau)\,.
$$
By theorem \ref{thm:base} and lemma \ref{lem:subvar}, the associated GVP is $((\d\theta-\d H\wedge\tau),\emptyset,\Gamma(V(\pi))$. 

If $R_t=\frac{\partial}{\partial t}$, then $\sigma_t=\d H-R_t(H)\tau$ and $\omega=\d\theta$. Since $\d\hat\theta$ is non-degenerate, there always exists a vector field $X\in\Gamma(V(\pi))$ such that $\inn_X\d\theta=\sigma_t\,.$
Therefore, the variational problem $(\Omega,\emptyset,\Gamma(V(\pi))$ is dynamically equivalent to $(\Omega,\emptyset,\mathfrak{X}(E))$.

The equations for a vector field $Z\in\mathfrak{X}(E)$ are (recall definition \ref{dfn:solvector}) 
\begin{align}
\inn_Z\tau=1\,,\quad \inn_Z\omega=\d H-R_t(H)\tau\,.   
\end{align}

They correspond to co-symplectic Hamiltonian dynamics \cite{munoz-lecanda_geometry_2024}. 

The form $\theta$ always projects to the $1$-form $\hat{\theta}$ of $T^*Q$ . If the Hamiltonian does not depend on $t$, both itself, and the form $\sigma_t$, also projects to $\hat{H}$ and $\hat{\sigma_t}$ on $T^*Q$, respectively. Then, a vector field solution is of the form $Z=R_t-X$, where $X$ is the pushforward by the $0$-section of a vector field $\hat{X}$ such that:
$$
\inn_{\hat{X}}\hat{\omega}=\d \hat{H}\,.
$$
Therefore, the solutions of the variational problem can be recovered from those of symplectic Hamiltonian dynamics.

\section{Nonholonomic Constraints}\label{sec:nonhol}
Suppose we have a set of constraints given by $1$-forms $I^{nh}_1\subset\Omega^1(E)$ and also constraints given by functions $I_0\subset C^\infty(E)$. We will denote the constraints by $(I_0,I^{nh}_1)$, instead of $J=I_0\cup I^{nh}_1$, to explicitly distinguish the different kinds of constraints. The nonholonomic implementation of $I^{nh}_1$ demands that the admissible variations must satisfy (See Definition 6 in \cite{gracia_geometric_2003})
\begin{align}
    \inn_\xi\eta\in I_0\,,\quad\forall\eta\in I^{nh}_1\,.
\end{align}

Together with the transversality constraint, we have that,
\begin{align}\label{eq:admNonholo}
    \textup{Adm}=\{\xi\in\mathfrak{X}(E)|\inn_\xi\tau=0, \inn_\xi\eta\in I_0,\,\forall \eta\in I_1\}.
\end{align}

Notice that Adm is $C^\infty(E)$-lineal. The objective of this section is, for any  $2$-form $\Omega$, to construct another $2$-form $\bar{\Omega}\in\Omega^2(E)$ such that
$(\Omega, (I_0,I_1), \textup{Adm})$ is dynamically equivalent to $(\bar{\Omega}, (I_0,I_1), \Gamma(V(\pi))$. To do so, some regularity will be assumed.

\begin{dfn}
    The nonholonomic constraints $I_1^{nh}$ are \textbf{co-oriented} if they are generated by a collection of $1$-forms $\{\eta^\alpha\}\subset\Omega^1(E)$ which, together with $\tau$, are linearly independent everywhere.
\end{dfn}

If $I_1^{nh}=\langle\eta^\alpha\rangle_{\alpha=1,\dots,k}$ is co-oriented, then there exists $R_t,R^\alpha\in\mathfrak{X}(E)$ such that:
\begin{align}\label{eq:reeb}
    \inn_{R_t}\tau=1,\quad
    \inn_{R_\alpha}\tau=0\,,\quad\inn_{ R_\alpha}\eta^\beta=\delta^\beta_\alpha\,, \text{ for all } \alpha,\beta=1,\dots,k\,.
\end{align}
Then, consider
$$
\bar{\Omega}=\Omega+\sigma_\alpha\wedge \eta^\alpha,
$$
with $\sigma_\alpha=\inn_{R_\alpha}\Omega$.

\begin{thm}\label{theo:nh}
    If $I_1^{nh}$ is co-oriented, the GVP $(\Omega, (I_0,I_1), \textup{Adm})$ is dynamically equivalent to the GVP $(\bar{\Omega}, (I_0,I_1), \Gamma(V(\pi))$.
\end{thm}
\begin{proof}
Let $\gamma\in\text{Sol}(\bar{\Omega}, (I_0,I_1), \Gamma(V(\pi))$, and consider a vector field $X\in\textup{Adm}$. Since $\textup{Adm}\subset\Gamma(V(\pi))$:
$$
0=\gamma^*(\inn_X\bar{\Omega})=\gamma^*(\inn_X\Omega+(\inn_X\sigma_\alpha)\eta^\alpha-(\inn_X\eta^\alpha)\sigma_\alpha)=\gamma^*(\inn_X\Omega)\,.
$$
Therefore,

    $$\text{Sol}(\Omega, (I_0,I_1), \textup{Adm})\supset\text{Sol}(\bar{\Omega}, (I_0,I_1), \Gamma(V(\pi))\,.
    $$
    Each vector field $Y\in\Gamma(V(\pi))$ can be decomposed as
    $$
    Y=X+\inn_Y\eta^\alpha R_\alpha,
    $$
    where $X=Y-\inn_Y\eta^\alpha R_\alpha$ has the properties $\inn_X\tau=0$ and $\inn_X\eta^\alpha=0$ for all $\alpha=1,\dots,k$.
    Let $\gamma\in\text{Sol}(\Omega, (I_0,I_1), \textup{Adm})$ and $Y\in\Gamma(V(\pi))$ :
    
\begin{align*}
 \gamma^*\inn_Y\bar{\Omega}&=\gamma^*[\inn_Y\Omega+(\inn_Y\sigma_\alpha)\eta^\alpha-\sigma_\alpha\inn_Y\eta^\alpha]
 \\
 &=\gamma^*[\inn_X\Omega+\inn_Y\eta^\alpha\inn_{R_\alpha}\Omega+(\inn_Y\sigma_\alpha)\eta^\alpha-\inn_Y\eta^\alpha\,\inn_{R_\alpha}\Omega ]=\gamma^*\inn_X\Omega=0.   
\end{align*}

Therefore, $\gamma\in\text{Sol}(\bar{\Omega}, (I_0,I_1), \Gamma(V(\pi))$.

\end{proof}
Notice that $\bar{\Omega}$ is not unique as it depends of the $R_\alpha$ chosen. As a corollary of the theorem, all possible elections of $R_\alpha$ lead to dynamically equivalent GVP.

\paragraph{Example: Hamiltonian (co)contact systems.} Consider $\pi:E=\mathbb{R}\times T^*Q\times\mathbb{R}\rightarrow \mathbb{R}$.  Let $\hat{\theta}$ be the canonical $1$-form of $T^*Q$ and denote by $\theta$ its pullback to $E$. In adapted local coordinates $(t,q^i,p_i,s)$ they have the expression:
$$
\theta=p_i\d q^i\,.
$$
Given a function $H:\mathbb{R}\times T^*Q\times\mathbb{R}\rightarrow \mathbb{R}$, the Hamiltonian version of the Herglotz variational problem consists on finding the critical sections $\gamma:\mathbb{R}\rightarrow E$ of the action
$$
\mathcal{A}(\mathcal{K},\gamma)=\int_{\mathcal{K}} \gamma^*(H\tau-\theta)
$$
with the nonholonomic constraint $\eta=H\d t+\d s-p_i\d q^i$. That is, $I_0=\emptyset$ and $I_1=\langle\eta\rangle$. The first to present the Herglotz's variational principle as a constrained variational problem were M.C. Muñoz-Lecanda and collaborators \cite{LLM-2021,LLM-2023}.

The solutions must satisfy $\gamma^*\eta=0$, and the admissible variations are
$$
\textup{Adm}=\{\xi\in\mathfrak{X}(E)|\inn_\xi\tau=0\,, \inn_\xi\eta=0\}
$$
Since $\Omega=\d H\wedge\tau-\d\theta$, the Hamiltonian Herglotz Variational problem is $(\d H\wedge\tau-\d\theta,(\emptyset,\langle\eta\rangle),\textup{Adm})$.
 
To find the associated geometry, we can take 
$$
R_t=\frac{\partial}{\partial t}\,,R_s=\frac{\partial}{\partial s}\,,
$$
that satisfy $\inn_{R_t}\tau=1\,,\inn_{R_s}\tau=0\,\text{ and }\inn_{R_s}\eta=1\,.$
Then, $\sigma_s=\inn_{R_s}\Omega=R_s(H)\d t$ and 
$$
\bar{\Omega}=\d H\wedge \d t-\d p_i\wedge\d q^i+R_s(H)\d t\wedge (H\d t+\d s-p_ i\d q^i).
$$
We can rewrite $\bar{\Omega}$ as
$$
\bar{\Omega}=\d \eta+\sigma_s\wedge \eta.
$$
By theorem \ref{theo:nh} the Hamiltonian Herglotz Variational problem is dynamically equivalent to $(\bar{\Omega},(\emptyset,\langle\eta\rangle),\Gamma(V(\pi)))$.

We can also reduce the transversality constraints as in section \ref{sec:trans}. The vector field $R_t=\frac{\partial}{\partial t}$ induces the decomposition
$$
\bar{\Omega}=\omega-\sigma_t\wedge\tau\,,
$$
with 
\begin{align}
 \sigma_t&=\inn_{R_t}\bar{\Omega}=-\d H+R_t(H)\d t+R_s(H)(\d s-p_i\d q^i)=-\frac{\partial H}{\partial p_i}\d p_i-\left(\frac{\partial H}{\partial q^i}+p_i\frac{\partial H}{\partial s}\right)\d q^i\,,
 \\
\omega&=\bar{\Omega}+\sigma_t\wedge\tau=-\d p_i\wedge \d q^i=-\d \theta\,.  
\end{align}
The vector field $X\in\Gamma(V(\pi))$
$$
X=-\frac{\partial H}{\partial p_i}\frac{\partial}{\partial q^i}+\left(\frac{\partial H}{\partial q^i}+p_i\frac{\partial H}{\partial s}\right)\frac{\partial}{\partial p_i}-\left(p_ i\frac{\partial H}{\partial p^i}-H\right)\frac{\partial}{\partial s}\,,
$$
satisfies:
$$
\inn_X\omega=\sigma_t\,, \quad \inn_{R_t-X}\eta=0\,.
$$
Therefore, the hypotheses of theorem \ref{theo:trans} hold and, hence, the variational problem \\ $(\bar{\Omega},(\emptyset,\langle\eta\rangle),\Gamma(V(\pi)))$ is dynamically equivalent to $(\bar{\Omega},(\emptyset,\langle\eta\rangle),\mathfrak{X}(E))$.

To better understand the geometry associated with this variational problem, we look at the equations for vector fields appearing in definition \ref{dfn:solvector}: 
\begin{align}
    \inn_{Z}\bar{\Omega}=0\,,\quad \inn_{Z}\tau=1\,,\quad \inn_{Z}\eta=0\,.
\end{align}

Denoting 
\begin{align}\label{eq:hateta}
    \hat{\eta}=\d s-p_i\d q^i\,,
\end{align}
it becomes apparent that they are the cocontact Hamiltonian equations:

\begin{align}
    \inn_{Z}\d \hat{\eta}&=\d H-R_s(H)\hat{\eta}-R_t(H)\tau\,,
    \\
    \quad \inn_{Z}\hat{\eta}&=-H\,,
    \\
    \inn_Z\tau&=1\,.
\end{align}
Thus, the Hamiltonian version of the Herglotz variational problem leads to cocontact geometry. See \cite{LGGMR-2022} for the definition of cocontact geometry and Hamiltonian equations.

The relation (\ref{eq:conservationLaw}) is
$$
\inn_Z\sigma_t=0\Rightarrow Z(H)=-HR_s(H)+R_t(H)\,,
$$
which coincides with the energy dissipation law of cocontact systems \cite{GRL-2023}.

\section{Vakonomic Constraints}\label{sec:vak}

A set of differential forms $J\subset\Omega^\bullet{E}$ is implemented as vakonomic constraints if the admissible variations $\xi$ leave $J$ invariant. Namely,
$$
\Lie_\xi\alpha\in J,\text{ for all }\alpha\in J\,.
$$
This implies that, if a section $\gamma$ satisfies $\gamma^*\alpha=0$ for $\alpha\in J$, then also $(\psi_s\circ\gamma)^*\alpha=0$, where $\psi_s$ is the element of the one-parameter group of diffeomorphisms of $E$ defined by the smooth vector field $\xi$ corresponding to parameter value $s$. This is the same approach used in \cite{gracia_geometric_2003} (See Definition 4 therein).


The main problem with vakonomic constraints is that the set of admissible variations in general is not $B(\pi)$-linear. This implies that one cannot take general enough variations to use the fundamental lemma of the calculus of variations. Here we will focus on two important exceptions: constraints generated by functions and the holonomy condition of Lagrangian systems.

\subsection{Constraints Generated by Functions}\label{subsec:functions}

 Let $J$ be the ideal generated by independent functions $f^i$: $I_0=\langle f^i\rangle$. Then, the admissible variations $\xi$ satisfy:
 $$
\Lie_\xi f^i=\inn_\xi\d f^i\in I_0\,.
 $$
 Clearly, if $\xi\in$Adm, for any function $g\in C^\infty(E)$, also $g\xi\in$Adm. Therefore, Adm is $C^\infty(E)$-linear and, in particular, $B(\pi)$-lineal.

 Moreover, the conditions of being admissible are equivalent to implementing the constraints $I_0=\langle f^i\rangle$ and $I_1=\langle \d f^i\rangle$ in a nonholonomic way. Therefore, we can use the same procedure as in section \ref{sec:nonhol}.

\subsection{The Holonomy Condition of Lagrangians Systems}\label{sec:holLag}
Consider the first jet bundle $\pi:E=J^1\rho\rightarrow \mathbb{R}$ of the fiber bundle $\rho:\mathbb{R}\times Q\rightarrow \mathbb{R}$, with $\dim Q=n$. Let $\mathcal{K}$ be a compact subset of $\mathbb{R}$. Given a function $L:J^1\rho\rightarrow \mathbb{R}$, Hamilton's variational principle consists on finding the critical holonomic local sections $\gamma$ of $\pi$ of the action
\begin{align}\label{var:Lag}
\mathcal{A}(\mathcal{K},\gamma)=\int_\mathcal{K} \gamma^*(L\tau)\,.    
\end{align}

The holonomy condition can be imposed by implementing the Cartan codistribution of $J^1\rho$ as vakonomic constraints.   In adapted local coordinates $(t,q^i,v^i)$ where $\tau=\d t$, the Cartan codistribution is generated by:
$$\kappa^i=\d q^i-v^i\d t\,.
$$

Therefore, $I_0=\emptyset$ and $I_1^{vak}=\langle\kappa^i\rangle_{i=1,\dots,n}$. Then, the solutions must be local section that satisfy 
$$
\gamma^*\kappa^i=0\,,\quad i=1,\dots,n
$$
and the admissible variations are
$$
\textup{Adm}=\{\xi\in\mathfrak{X}(E)|\inn_\xi\tau=0\,, \Lie_\xi\kappa^i\in I_1^{vak}\,, i=1,\dots,n\}\,.
$$

Adm is not $B(\pi)$-linear. The classical derivation of the Euler-Lagrange equations manipulates the variation of the action with respect to admissible variations so that the fundamental lemma of the calculus of variations applies. A more geometric presentation of these computations shows that a modification of the form $L\tau$ is required. The Lepage equivalent theory \cite{Krupka} studies the different differential forms leading to the same differential equations in the presence of holonomy constraints generated by the Cartan codistribution. For completeness, we present here an alternative derivation.

Define the distribution
$$
\mathcal{C}=\bigcap_{i=1,\dots,n}\ker \kappa^i\cap\tau\,.
$$
In local coordinates, the sections of the distribution $\mathcal{C}$ are generated by:
$$
\Gamma(\mathcal{C})=\left\langle\frac{\partial}{\partial v^i}\right\rangle_{i=1,\dots,n}\,.
$$

Consider the differential form
\begin{align}\label{eq:PoincareCartan}
\Theta_L=L\d t+\frac{\partial L}{\partial v^i}\kappa^i\,.    
\end{align}

The $1$-form $\Theta_L$ is the \textbf{Poincaré-Cartan form}  \cite{PLG-1974}, that can be intrinsically defined using the canonical structures of the jet bundle \cite{echeverria-enriquez_geometry_1996,munoz-lecanda_geometry_2024}. It has the additional property that, if $Y\in\Gamma(C)$, and $\gamma$ is a holonomic section, then
\begin{align}
 \gamma^*\left[\Lie_Y\Theta_L\right]&
 =\gamma^*\left[\Lie_Y(L\d t)+\Lie_Y\left(\frac{\partial L}{\partial v^i}\right)\kappa^i+\frac{\partial L}{\partial v^i}\inn_Y\d\kappa^i\right]
 \\
 &=\gamma^*\left[\left(\Lie_YL-\frac{\partial L}{\partial v^i}\inn_Y\d v^i\right)\d t+\Lie_Y\left(\frac{\partial L}{\partial v^i}\right)\kappa^i\right]
 =\gamma^*\left[\Lie_Y\left(\frac{\partial L}{\partial v^i}\right)\kappa^i\right]
 =0\,.   
\end{align}

\begin{prop}
    The VP $(L\tau,(\emptyset,I_1^{vak}),\textup{Adm})$ is dynamically equivalent to the VP
    \\$(\Theta_L,(\emptyset,I_1^{vak}),\langle\textup{Adm}\rangle^{B(\pi)})$.
\end{prop}
\begin{proof}

For any section $\gamma$ satisfying the constraints and for any $\xi\in$ Adm:
$$
\gamma^*\left[\Lie_\xi\Theta_L\right]=\gamma^*\left[\Lie_\xi(L\d t)+\Lie_\xi\left(\frac{\partial L}{\partial v^i}\kappa^i\right)\right]=\gamma^*[\Lie_\xi(L\d t)]\,.
$$

Therefore, $Sol(\Theta_L,(\emptyset,I_1^{vak}),\langle\textup{Adm}\rangle^{B(\pi)})\subset Sol(L\tau,(\emptyset,I_1^{vak}),\textup{Adm})$.

If $\xi\in$ Adm, its local expression is:
$$
f^i(t,q^i)\frac{\partial}{\partial q^i}+D_tf^i\frac{\partial}{\partial v^i}\,,
$$
where $D_t=\frac{\partial}{\partial t}+v^i\frac{\partial}{\partial q^i}$ is the total derivative. Given a $\pi$-basic function $g(t)\in B(\pi)$, the vector field $g\xi\notin$ Adm  in general. Nevertheless, there exists a vector field $Y_{g,\xi}\in\Gamma(\mathcal{C})$ such that $g\xi+Y_{g,\xi}\in$ Adm . Indeed, if 
$$
Y_{g,\xi}=f^i\frac{\d g}{\d t}\frac{\partial}{\partial v^i}\,,
$$
then
\begin{align*}
   g\xi+Y_{g,\xi}=gf^i\frac{\partial}{\partial q^i}+\left(gD_tf^i+f^i\frac{\d g}{\d t}\right)\frac{\partial}{\partial v^i}=gf^i\frac{\partial}{\partial q^i}+D_t\left(gf^i\right)\frac{\partial}{\partial v^i} \,.
\end{align*}

Therefore, for a holonomic section $\gamma$,
\begin{align*}
\gamma^*\left[\Lie_{g\xi}\Theta_L\right]
=\gamma^*\left[\Lie_{g\xi+Y_{g,\xi}}\Theta_L-\Lie_{Y_{g,\xi}}\Theta_L\right]
=\gamma^*\left[\Lie_{g\xi+Y_{g,\xi}}\Theta_L\right]=\gamma^*[\Lie_{g\xi+Y_{g,\xi}}(L\d t)]\,.
\end{align*}

Consequently,
 $Sol(L\tau,(\emptyset,I_1^{vak}),\textup{Adm})\subset Sol(\Theta_L,(\emptyset,I_1^{vak}),\langle\textup{Adm}\rangle^{B(\pi)})$.
    
\end{proof}

As a corolary of theorem \ref{thm:base}, we have that the VP $(\Theta_L,(\emptyset,I_1^{vak}),\langle\textup{Adm}\rangle^{B(\pi)})$ is dynamically equivalent to the GVP $(\d\Theta_L,(\emptyset,I_1^{vak}),\langle\textup{Adm}\rangle^{B(\pi)})$. Its equations of motion are
\begin{align}\label{eq:EL}
    \gamma^*\inn_{\xi}\d\Theta_L=0\,,\quad\forall \xi\in\langle\textup{Adm}\rangle^{B(\pi)}
\end{align}

which are the Euler-Lagrange equations.

The geometric equations \eqref{eq:EL} are not in the typical form one finds in the literature, where there is an equation for each vertical vector field, not only admissible (see, for instance, \cite{echeverria-enriquez_geometry_1996,munoz-lecanda_geometry_2024}). Both formulations are dynamically equivalent.

\begin{prop}
    The GVP $(\d \Theta_L,(\emptyset,I_1^{vak}),\langle\textup{Adm}\rangle^{B(\pi)})$ is dynamically equivalent to the GVP 
    \\$(\d \Theta_L,(\emptyset,I_1^{vak}),\Gamma(V(\pi)))$.
\end{prop}
\begin{proof}
    Since $\langle\textup{Adm}\rangle^{B(\pi)}\subset\Gamma(V(\pi))$, being the other elements equal, by lemma \ref{lem:subvar} we have that
    $$
\text{Sol}(\d \Theta_L,(\emptyset,I_1^{vak}),\langle\textup{Adm}\rangle^{B(\pi)})\supset\text{Sol}(\d \Theta_L,(\emptyset,I_1^{vak}),\Gamma(V(\pi)))\,.
    $$
The local expression of a generic vertical vector field $\Gamma(V(\pi))$ is
$$
f^i\frac{\partial}{\partial q^i}+g^i\frac{\partial}{\partial v^i}\,;\quad f^i,g^i\in C^{\infty}(E)
$$
Let $\gamma$ be a solution of $(\d \Theta_L,(\emptyset,I_1^{vak}),\langle\textup{Adm}\rangle^{B(\pi)})$. The equation of motion of $(\d \Theta_L,(\emptyset,I_1^{vak}),\Gamma(V(\pi)))$ are:
$$
0=\gamma^*\inn_\xi\d \Theta_L=f^i(\gamma)\gamma^*\inn_{\frac{\partial}{\partial q^i}}\d \Theta_L\Leftrightarrow \gamma^*\inn_{\frac{\partial}{\partial q^i}}\d \Theta_L=0\,,\forall i=1,\dots,n.
$$
The vector fields $\frac{\partial}{\partial q^i}$ are elements of $\langle\textup{Adm}\rangle^{B(\pi)}$. Therefore, the equation of motion generated by $\langle\textup{Adm}\rangle^{B(\pi)}$ are the same as those generated by $\Gamma(V(\pi))$.
\end{proof}

Next, we reduce the transversality condition using theorem \ref{theo:trans}. To streamline the exposition, we introduce two common notions in  Lagrangian mechanics~\cite{munoz-lecanda_geometry_2024}. The Lagrangian energy is
\begin{align}\label{eq:lagrangianenergy}
E_L=\frac{\partial L}{\partial v^i}v^i-L\,.
\end{align}
A Lagrangian $L$ is regular if 
$$
\frac{\partial^2 L}{\partial v^i\partial v^j}\,,
$$
is a regular matrix everywhere. Then, there exists a matrix $W^{ij}\in C^\infty(E)$ which is its inverse in the sense that, at each point:
    $$
W^{ij}\frac{\partial^2 L}{\partial v^j\partial v^k}=\delta^i_k\,.
    $$
\begin{cor}\label{cor:lag}
    If $L$ is regular, the GVP $(\d \Theta_L,(\emptyset,I_1^{vak}),\Gamma(V(\pi)))$ is dynamically equivalent to the GVP $(\d \Theta_L,(\emptyset,I_1^{vak}),\mathfrak{X}(E))$.
\end{cor}
\begin{proof}    
    Choose 
   \begin{align}\label{reeb:vak}
           R_t=\frac{\partial}{\partial t}-W^{ij}\frac{\partial^2L}{\partial v^i\partial t}\frac{\partial}{\partial v^j}\,.
   \end{align}
    Then, the forms defined in \eqref{eq:formstrans} are
   \begin{align}
    \sigma_t&=\d E_L-R_t(E_L)\tau\,,
    \\
    \omega&=\d \frac{\partial L}{\partial v^i}\wedge\d q^i\,.
\end{align}

The vector field
$$
X=-v^i\frac{\partial}{\partial q^i}-W^{ij}\left(\frac{\partial L}{\partial q^j}-v^k\frac{\partial^2 L}{\partial q^k\partial v^j}\right)\frac{\partial}{\partial v^i}
$$
satisfies:
$$
\inn_X\omega=\sigma_t\,,\quad\inn_{R_t-X}\tau=1\,,\quad\inn_{R_t-X}\kappa^i=0\,,\; i=1,\dots,n\,.
$$
Therefore, the hypotheses of theorem \ref{theo:trans} hold.

\end{proof}
From the same $\Theta_L$, one can choose different $R_t$ vector fields, with a corresponding $\omega$ (recall its definition in equations \eqref{eq:formstrans}), leading to the same equations. The $R_t$ in the corollary \ref{cor:lag} has been chosen so that $\omega$ coincides with (minus) the Lagrange $2$-form associated with $L$ \cite{munoz-lecanda_geometry_2024}.

To better understand the geometry associated with this variational problem, we look at the equations for vector fields: 
\begin{align}
    \inn_{Z}\d \Theta_L=0\,,\quad \inn_{Z}\tau=1\,,\quad \inn_{Z}\kappa^i=0\,.
\end{align}

The last two conditions are usually stated as demanding that the vector field $Z$ is holonomic. The first one is equivalent to the cosymplectic Lagrangian equations \cite{munoz-lecanda_geometry_2024}:

\begin{align}
    \inn_{Z}\omega&=\d {E_L}-R_t(E_L)\tau\,,
\end{align}
Thus, the Hamilton variational principle \eqref{var:Lag} for regular Lagrangians leads to the cosymplectic structure $(\omega,\tau)$. The different $\omega$ resulting from a different choice of $R_t$ generate different cosymplectic structures with the same dynamics. 

The relation (\ref{eq:conservationLaw}) is
$$
\inn_Z\sigma_t=0\Rightarrow Z(E_L)=R_t(H)\,,
$$
which coincides with the energy dissipation law for cosymplectic systems.

\section{Nonholonomic and Vakonomic Constraints}\label{sec:nhvak}

The main motivation of this work is to find the geometry for the Lagrangian Herglotz variational problem, which combines nonholonomic and vakonomic constraints. It has the convenient property that the constraints are ``compatible''. A similar compatibility condition will be used to adapt the arguments of section \ref{sec:vak}. 

Consider the first jet bundle $J^1\rho$ of the fiber bundle $\rho:\mathbb{R}\times Q\rightarrow \mathbb{R}$, with dim$Q=n$. Let $E=J^1\rho\times N$, for a manifold $N$ of dim$N=k+r$\,. We denote the projections $\pi:E\rightarrow \mathbb{R}$, $\pi_N:E\rightarrow N$ and $\pi_{j^1\rho}:E\rightarrow J^1\rho$. Let $\hat{\tau}$ be a volume form on $\mathbb{R}$ and denote by $\tau$ its pullback to $E$. We will use local adapted coordinates $(t,q^i,v^i,s^\alpha,u^a)$, with $\d t=\tau$, $i=1,\dots,n$, $\alpha=1,\dots,k$ and $a=1,\dots,r$. $(t,q^i,v^i)$ is a local chart of $J^1\rho$ and $(s^\alpha,u^a)$ of $N$ that we will make more precise later \eqref{eq:formscoordinates}.

\begin{equation*}
\xymatrix{
&\   &E=J^1\rho\times N \ar[rd]^{\pi_N}\ar[ld]_{\pi_{j^1\rho}}\ar[lddd]_{\pi}  \   &
\\
&J^1\rho \ar[d]\ &  \ &N
\\
&\mathbb{R}\times Q\ar[d]^{\rho}\ & \ &  
 \\
&\mathbb{R}\ &  \ & 
}
\end{equation*}

The vakonomic constraints are the Cartan codistribution of $J^1\rho$, pulled back to $E$, $I_1^{vak}=\langle \kappa^i\rangle$. In adapted local coordinates:
$$\kappa^i=\d q^i-v^i\d t\,.
$$

The nonholonomic constraints are $I_1^{nh}\subset\Omega^1(E)$. Moreover, we consider constraints generated by functions $I_0=\langle f_j\rangle_{j=1,\dots,m}\subset C^\infty(E)$. They can be equivalently incorporated nonholonomically or vakonomically, as explained in \ref{subsec:functions}. We will use the former, that is $\d f_j\in I^{nh}_1$.

The admissible variations are
\begin{align}\label{var:nhvak}
\textup{Adm}=\{\xi\in\mathfrak{X}(E)|\,\inn_\xi\tau=0,\,\inn_\xi\eta\in I_0,\,\Lie_ \xi\kappa\in \langle I_1^{vak}, I_1^{nh}\rangle;\,\forall \eta\in I^{nh}_1\,,\forall \kappa\in I^{vak}_1\}\,.
\end{align}

Here $ \langle I_1^{vak}, I_1^{nh}\rangle$ denotes all linear combinations of elements of $I_1^{vak}$ and/or $I_1^{nh}$, with coefficients in $C^\infty(E)$. The condition $\Lie_ \xi\kappa\in  \langle I_1^{vak}, I_1^{nh}\rangle$ has been chosen following the idea that objects vanish up to constraints.


Given a $1$-form $\theta\in\Omega^1(E)$, the VP is $(\theta,(I_0,I_1^{nh},I_1^{vak}),\textup{Adm})$. The main obstacle to deriving a geometry which induces the same dynamics as this variational problem is that Adm is not $B(\pi)$-linear. With some compatibility conditions that we are going to introduce next, one can apply the fundamental lemma of the calculus of variations, as and open the way to its geometrization.

Define the distribution:
$$
\mathcal{C}=\bigcap_{p\in E}\mathcal{C}_p=\bigcap_{p\in E}\{v\in T_p E\,|\,T_p\pi_N(v)=0,\tau_p(v)=0,\kappa_p^i(v)=0\,,i=1,\dots,n\}\,.
$$
In local coordinates, the sections of the distribution $\mathcal{C}$ are generated by:
$$
\Gamma(\mathcal{C})=\left\langle\frac{\partial}{\partial v^i}\right\rangle_{i=1,\dots,n}\,.
$$
\begin{dfn} The constraints $I^{nh}_1$ and $I^{vak}_1=\langle \kappa^i\rangle$ are \textbf{compatible} if:
\begin{enumerate}
    \item $I^{nh}_1$ is generated by $k$ $1$-forms: $I^{nh}_1=\langle \eta^\alpha\rangle_{\alpha=1,\dots,k}$.
    \item The forms $\eta^\alpha$, together with $\tau$, $\kappa^i$ and $\d \kappa^i$, are linearly independent with respect to the other forms:
    \begin{align*}        \tau\wedge\left(\bigwedge_{\alpha=1}^k\eta^\alpha\right)\wedge\left(\bigwedge_{i=1}^n\kappa^i\wedge\d\kappa^i\right)\neq0\,.
    \end{align*}
    \item $\displaystyle\mathcal{C}\subset \bigcap_{\eta\in I^{nh}_1}\ker\eta$   
\end{enumerate}    
\end{dfn}

If the constraints are compatible, then there exist $k$ vector fields $R_\alpha\in\mathfrak{X}(E)$ such that 
\begin{align}\label{eq:reebnhvak}
\inn_{R_\alpha}\tau=0\,,\quad \inn_{R_\alpha}\eta^\beta=\delta^\beta_\alpha\,.
    \end{align}
    
We also have adapted coordinates $(t,q^i,v^i,s^\alpha, u^a)$ such that:
\begin{align}\label{eq:formscoordinates}
    \tau=\d t\,,\quad \kappa^i=\d q^i-v^i\d t\,,\quad \eta^\alpha= \d s^\alpha+b_i^\alpha \d q^i+h^\alpha\d t +c^\alpha_a\d u^a\,,
\end{align}
for some functions $b^\alpha_ i,h^\alpha, c^\alpha_a\in C^\infty(E)$.

The condition $2$ precludes the presence of nontrivial constraints $f\in I_0$ projectable to $J^1\rho$. 

The construction of the geometry associated to the variational problem $(\theta,(I_0,I_1^{nh},I_1^{vak}),\textup{Adm})$ will be done in several steps, summarized in diagram \ref{fig:concatenation}.

\begin{figure}[H]\label{fig:concatenation}
\begin{equation*}
\xymatrix{
&(L\tau,(I_0,I_1^{nh},I_1^{vak}),\textup{Adm})\ar@{<->}[rrr]^{\text{Theorem }\ref{thm:vaknh1} }
\ &
\ &
\ &
(\Theta_L,(I_0,I_1^{nh},I_1^{vak}),\langle\textup{Adm}\rangle^{B(\pi)})\ar@{<->}[d]^{\text{Theorem }\ref{thm:base} }
\\
&(\overline\Omega,(I_0,I_1^{nh},I_1^{vak}),\Gamma(V(\pi)))\ar@{<->}[rrr]^{\text{Theorem }\ref{thm:vaknh2} }
\ &  
\ &
\ &(\d\Theta_L,(I_0,I_1^{nh},I_1^{vak}),\langle\textup{Adm}\rangle^{B(\pi)})
}
\end{equation*}\caption{Concatenation of dynamically equivalent systems}
\end{figure}

\begin{thm}\label{thm:vaknh1}
    Consider a VP $(\theta,(I_0,I_1^{nh},I_1^{vak}),\textup{Adm})$ such that
    \begin{enumerate}
        \item $\theta=L\tau$ for some function $L\in C^\infty(E)$.
        \item The constraints are compatible.
    \end{enumerate}
    Then, the VP $(\theta,(I_0,I_1^{nh},I_1^{vak}),\textup{Adm})$ is dynamically equivalent to the VP  $(\Theta_L,(I_0,I_1^{nh},I_1^{vak}),\langle\textup{Adm}\rangle^{B(\pi)})$,
    with \footnote{The form $\Theta_L$ can be defined intrinsically using the canonical structures of the jet bundle $J^1\rho$ pulled back to $E$ \cite{munoz-lecanda_geometry_2024}.}
    \begin{align}\label{eq:thetalvaknh}
        \Theta_L=L\tau+\frac{\partial L}{\partial v^i}\kappa^i\,.
    \end{align}
\end{thm}

\begin{proof}
For any section $\gamma$ that satisfies the constraints, and for any vector field $\xi\in$ Adm:
\begin{align}
    \gamma^*\left[\Lie_\xi \left(L \d t+\frac{\partial L}{\partial v^i}\kappa^i\right)\right]=\gamma^*\left[ \Lie_\xi L \d t+\left(\Lie_\xi\frac{\partial L}{\partial v^i}\right)\kappa^i+\frac{\partial L}{\partial v^i}\Lie_\xi\kappa^i\right]=\gamma^*\Lie_\xi \left(L \d t\right)\,.
\end{align}
Therefore, $Sol(\Theta_L,(I_0,I_1^{nh},I_1^{vak}),\langle\textup{Adm}\rangle^{B(\pi)})\subset Sol(\theta,(I_0,I_1^{nh},I_1^{vak}),\textup{Adm})$.

Consider adapted coordinates $(t,q^i,v^i,s^\alpha, u^a)$ as in \ref{eq:formscoordinates}. A vertical vector field $\xi\in\Gamma(V(\pi))$ has local expression
    $$
\xi=f^i\frac{\partial}{\partial q^i}+F^i\frac{\partial}{\partial v^i}+f^\alpha\frac{\partial}{\partial s^\alpha}+f^a\frac{\partial}{\partial u^a}\,.
    $$
If it is in Adm, we have that $\inn_\xi\eta^\alpha\in I_0$ which implies:
\begin{align}\label{eq:adm1}
    f^\alpha=-b^\alpha_if^i-c^\alpha_af^a+T\,,
\end{align}
For some $T\in I_0$. Moreover, 
\begin{align}
    \Lie_ \xi\kappa^i&=\frac{\partial f^i}{\partial s^\alpha}\eta^\alpha
    +\left[ \frac{\partial f^i}{\partial q^j}-b^\alpha_j\frac{\partial f^i}{\partial s^\alpha}\right]\kappa^j
    +\frac{\partial f^i}{\partial v^j}\d v^j
    +\left[ \frac{\partial f^i}{\partial u^a}-c^\alpha_a\frac{\partial f^i}{\partial s^\alpha}\right]\d u^a
    \\
    &+\left[\frac{\partial f^i}{\partial t}+v^j\frac{\partial f^i}{\partial q^j}-\left(b^\alpha_jv^j+h^\alpha\right)\frac{\partial f^i}{\partial s^\alpha}-F^i\right]\d t\,.
\end{align}

If $\Lie_ \xi\kappa^i\in \langle I_1^{vak}, I_1^{nh}\rangle$, then
\begin{align}\label{eq:adm2}
    \frac{\partial f^i}{\partial v^j}=0\,, \quad \frac{\partial f^i}{\partial u^a}=c^\alpha_a\frac{\partial f^i}{\partial s^\alpha}\,,\quad
    F^i=\frac{\partial f^i}{\partial t}+v^j\frac{\partial f^i}{\partial q^j}-\left(b^\alpha_jv^j+h^\alpha\right)\frac{\partial f^i}{\partial s^\alpha}\,.
\end{align}

Equations \eqref{eq:adm1} and \eqref{eq:adm2} are the local conditions a vector field must satisfy to be in Adm. If $\xi\in$ Adm and $g\in B(\pi)$, then $g\xi\notin\textup{Adm}$ in general. Nevertheless, there exists a vector field $Y_{g,\xi}\in\Gamma(\mathcal{C})$ such that $g\xi+Y_{g,\xi}\in\textup{Adm}$. Indeed, if
\begin{align}
Y_{g,\xi}=f^i\frac{\d g}{\d t}\frac{\partial}{\partial v^i}
\end{align}

then,
\begin{align*}
    \inn_{g\xi+Y_{g,\xi}}\d t=0\quad \text{and}\quad  \inn_{g\xi+Y_{g,\xi}}\eta^\alpha=\inn_{g\xi}\eta^\alpha\in I_0\,\quad \text{forall }\alpha\,.
\end{align*}
Moreover, 

\begin{align*}
    \Lie_{g\xi+Y_{g,\xi}}\kappa^i=(\inn_{\xi}\kappa^i)\d g+g\Lie_\xi\kappa^i+ \Lie_{Y_{g,\xi}}\kappa^i=f^i\frac{\partial g}{\partial t}\d t-f^i\frac{\d g}{\d t}\d t+g\Lie_\xi\kappa^i\in \langle I_1^{vak},I_1^{nh}\rangle\,,\text{ for all } i\,.
\end{align*}

For a section $\gamma$ satisfying the constraints,
\begin{align*}
\gamma^*\left[\Lie_{g\xi}\Theta_L\right]
=\gamma^*\left[\Lie_{g\xi+Y_{g,\xi}}\Theta_L-\Lie_{Y_{g,\xi}}\Theta_L\right]
=\gamma^*\left[\Lie_{g\xi+Y_{g,\xi}}\Theta_L\right]=\gamma^*[\Lie_{g\xi+Y_{g,\xi}}(L\d t)]\,.
\end{align*}

Consequently, $Sol(\theta,(I_0,I_1^{nh},I_1^{vak}),\textup{Adm})\subset Sol(\Theta_L,(I_0,I_1^{nh},I_1^{vak}),\langle\textup{Adm}\rangle^{B(\pi)})$.
\end{proof}

 As a corolary of theorem \ref{thm:base}, we have that the VP $(\Theta_L,(I_0,I_1^{nh},I_1^{vak}),\langle\textup{Adm}\rangle^{B(\pi)})$ is dynamically equivalent to the GVP $(\d\Theta_L,(I_0,I_1^{nh},I_1^{vak}),\langle\textup{Adm}\rangle^{B(\pi)})$. Its equations of motion are
\begin{align}
    \gamma^*\inn_{\xi}\d\Theta_L=0\,,\quad\forall \xi\in\langle\textup{Adm}\rangle^{B(\pi)}\,.
\end{align}

After proving that a geometric version of $(\theta,(I_0,I_1^{nh},I_1^{vak}),\textup{Adm})$ exists, we proceed to extend it to more general variations.

\begin{thm}\label{thm:vaknh2}
    If the constraints are compatible, the GVP $(\d \Theta_L,(I_0,I_1^{nh},I_1^{vak}),\langle\textup{Adm}\rangle^{B(\pi)})$ is dynamically equivalent to the GVP $(\bar{\Omega},(I_0,I_1^{nh},I_1^{vak}),\Gamma(V(\pi))$, where
\begin{align}
\bar{\Omega}=\d\Theta_L+\sigma_\alpha\wedge \eta^\alpha,
\end{align}
with $\sigma_\alpha=\inn_{R_\alpha}\d\Theta_L$ and $R_\alpha$ satisfying \eqref{eq:reebnhvak}.
\end{thm}
\begin{proof}
Let $\gamma\in\text{Sol}(\bar{\Omega},(I_0,I_1^{nh},I_1^{vak}),\Gamma(V(\pi))$, and consider a vector field $\xi\in\langle\textup{Adm}\rangle^{B(\pi)}$. Since $\langle\textup{Adm}\rangle^{B(\pi)}\subset\Gamma(V(\pi))$:
$$
0=\gamma^*\inn_\xi\bar{\Omega}=\gamma^*\left[\inn_\xi\d \Theta_L+(\inn_\xi\sigma_\alpha)\eta^\alpha-(\inn_\xi\eta^\alpha)\sigma_\alpha\right]=\gamma^*\inn_\xi\d \Theta_L\,.
$$
Therefore,

    $$\text{Sol}(\d \Theta_L,(I_0,I_1^{nh},I_1^{vak}),\langle\textup{Adm}\rangle^{B(\pi)})\supset\text{Sol}(\bar{\Omega},(I_0,I_1^{nh},I_1^{vak}),\Gamma(V(\pi))\,.
    $$
    Each vector field $Y\in\Gamma(V(\pi))$ can be decomposed as
    $$
    Y=X+\inn_Y\eta^\alpha R_\alpha,
    $$
    where $X=Y-\inn_Y\eta^\alpha R_\alpha$ has the properties $\inn_X\tau=0$ and $\inn_X\eta^\alpha=0$ for all $\alpha=1,\dots,k$. A vector field with these properties has the local expression (using \eqref{eq:formscoordinates} for the local expression of $\eta^\alpha$):
\begin{align}
X&=f^i\frac{\partial}{\partial q^i}+F^i\frac{\partial}{\partial v^i}-\left(b^\alpha_if^i+c^\alpha_ af^a\right)\frac{\partial}{\partial s^\alpha}+f^a\frac{\partial}{\partial u^a}
\\
&=f^i\left(\frac{\partial}{\partial q^i}-b^\alpha_i\frac{\partial}{\partial s^\alpha}\right)+f^a\left(\frac{\partial}{\partial u^a}-c^\alpha_a\frac{\partial}{\partial s^\alpha}\right)+F^i\frac{\partial}{\partial v^i}\,.
\end{align}

Therefore, it is a linear combination of vector fields in Adm plus an element of $\Gamma(\mathcal{C})$. Namely, 
\begin{align}
    X=G_\mu\xi^\mu+\chi\,, 
\end{align}
for some $\xi^\mu\in$ Adm and $G_\mu\in C^\infty(E)$, with $\chi\in\Gamma(\mathcal{C})$. 
    Let $\gamma\in\text{Sol}(\d \Theta_L,(I_0,I_1^{nh},I_1^{vak}),\langle\textup{Adm}\rangle^{B(\pi)})$ and $Y\in\Gamma(V(\pi))$ :   
\begin{align}
 \gamma^*\inn_Y\bar{\Omega}&=\gamma^*[\inn_Y\d\Theta_L+(\inn_Y\sigma_\alpha)\eta^\alpha-\sigma_\alpha\inn_Y\eta^\alpha]
 \\
 &=\gamma^*[\inn_X\d\Theta_L+\inn_Y\eta^\alpha\inn_{R_\alpha}\d\Theta_L+(\inn_Y\sigma_\alpha)\eta^\alpha-\inn_Y\eta^\alpha\,\inn_{R_\alpha}\d\Theta_L ]=\gamma^*\inn_X\d\Theta_L
 \\
 &=\gamma^*\left[G_\mu\inn_{\xi^\mu}\d\Theta_L+\inn_\chi\d\Theta_L\right]
 \\ &=\left(\gamma^*G_\mu\right)\gamma^*\left[\inn_{\xi^\mu}\d\Theta_L\right]+\gamma^*\left[\chi(L)\d t+\chi\left(\frac{\partial L}{\partial v^i}\right)\kappa^i+\frac{\partial L}{\partial v^i}\inn_\chi\d \kappa^i\right]
 \\
 &=0.   
\end{align}

Therefore, $\gamma\in\text{Sol}(\bar{\Omega},(I_0,I_1^{nh},I_1^{vak}),\Gamma(V(\pi))$.

\end{proof}

\paragraph{Example: Regular action-dependent Lagrangians}\label{sec:regularlag}  Consider $E=J^1\rho\times\mathbb{R}$, with coordinates $(t,q^i,v^i,s)$. A (non-autonomous) action-dependent Lagrangian is a function $L:E\rightarrow \mathbb{R}$. Let $\mathcal{K}$ be a compact subset of $\mathbb{R}$. The Herglotz variational problem consists in finding holonomic critical sections $\gamma:I_\gamma\subset\mathbb{R}\rightarrow E$ of the action 

$$
\mathcal{A}(\mathcal{K},\gamma)=\int_\mathcal{K} \gamma^*(L\d t)
$$
which satisfy the nonholonomic constraint 
\begin{align}\label{eq:etanhvak}
    \eta=\d s+\left(\frac{\partial L}{\partial v^i}v^i-L\right) \d t-\frac{\partial L}{\partial v^i}\d q^i\,.
\end{align}
This constraint imposes the condition that the variable $s$ is the flow of action. In this case we have the vakonomic constraints given by the Cartan codistribution $I^{vak}_1=\langle\kappa^i\rangle$. We also have the nonholonomic constraint $I^{nh}_1=\langle\eta\rangle$.  The admissible variations are 
\begin{align}
    \textup{Adm}=\{\xi\in\mathfrak{X}(E)|\;\inn_\xi\tau=0\,, \Lie_Y\kappa\in \langle I^{vak}_1, I^{nh}_1\rangle\,,\forall\kappa\in I^{vak}_1\,, \inn_Y\eta=0\}
\end{align}

The constraints are compatible because $\inn_{\frac{\partial}{\partial v^i}}\eta=0$ for $i=1,\dots,n$, and 

\begin{align*}        \tau\wedge\left(\bigwedge_{\alpha=1}^k\eta^\alpha\right)\wedge\left(\bigwedge_{i=1}^n\kappa^i\wedge\d\kappa^i\right)=(-1)^n       \d t\wedge\left(\bigwedge_{\alpha=1}^k\d s^\alpha\right)\wedge\left(\bigwedge_{i=1}^n\d q^i\wedge\d v^i\right)\neq0\,.
    \end{align*}

Is it possible to find a vector field such that $\inn _{R_s}\eta=1$ and $\inn_{R_s}\tau=0$. For instance,
\begin{align}\label{eq:reeblagcoco}
R_s=\frac{\partial}{\partial s}-W^{ij}\frac{\partial^2L}{\partial s\partial v^j}\frac{\partial}{\partial v^i}\,,
\end{align}
which exists because $L$ is regular. Therefore, by theorem \ref{thm:vaknh1} and \ref{thm:base}, the Herglotz variational problem $(L\d t,(\emptyset,I_1^{nh},I_1^{vak}),\textup{Adm})$ has the same solutions as $(\d \Theta_L,(\emptyset,I_1^{nh},I_1^{vak}),\langle\textup{Adm}\rangle^{B(\pi)})$, with $\d \Theta_L$ proposed in \eqref{eq:thetalvaknh}.

With the choice of $R_s$ given in \eqref{eq:reeblagcoco} we have 
\begin{align}
   \sigma_s=\frac{\partial L}{\partial s}\d t\,,\quad{\bar{\Omega}}=\d \Theta_L+\sigma_s\wedge\eta\,.
\end{align}
Notice that $\d\Theta_L=-\d \eta$ and, therefore, ${\bar{\Omega}}=-\left(\d \eta -\sigma_s\wedge\eta\right)$. According to theorem \ref{thm:vaknh2}, the variational problem $(\d \Theta_L,(\emptyset,I_1^{nh},I_1^{vak}),\langle\textup{Adm}\rangle^{B(\pi)})$ is dynamically equivalent to $(\bar{\Omega},(\emptyset,I_1^{nh},I_1^{vak}),\Gamma(V(\pi))$.

Choose 
\begin{align}
    R_t=\frac{\partial}{\partial t}-W^{ij}\frac{\partial^2L}{\partial v^i\partial t}\frac{\partial}{\partial v^j}\,.
\end{align}
We have the identities
\begin{align}
    R_s(E_L)=-\frac{\partial L}{\partial s}\,,\quad R_t(E_L)=-\frac{\partial L}{\partial t} \,.
\end{align}
Then, the forms defined in \eqref{eq:formstrans} are
   \begin{align}
    \sigma_t&=\d E_L-R_t(E_L)\d t-R_s(E_L)\hat{\eta}\,.
    \\
    \omega&=\d \frac{\partial L}{\partial v^i}\wedge\d q^i\,.
\end{align}
where $\hat{\eta}=\d s-\frac{\partial L}{\partial v^i}\d q^i$.
The vector field
\begin{align}
X=-v^i\frac{\partial}{\partial q^i}-W^{ij}\left(\frac{\partial L}{\partial q^j}-v^k\frac{\partial^2 L}{\partial q^k\partial v^j}-L\frac{\partial^2L}{\partial s\partial v^i}+\frac{\partial L}{\partial s}\frac{\partial L}{\partial v^i}\right)\frac{\partial}{\partial v^i}-L\frac{\partial}{\partial s}
\end{align}
satisfies:
\begin{align}
    \inn_X\omega=\sigma_t\,,\quad\inn_{R_t-X}\tau=1\,,\quad \inn_{R_t-X}\eta=0\,,\quad \inn_{R_t-X}\kappa^i=0\,,\quad i=1,\dots,n\,.
\end{align}

Therefore, the hypotheses of theorem \ref{theo:trans} hold and the variational problem is dynamically equivalent to $(\bar{\Omega},(\emptyset,I_1^{nh},I_1^{vak}),\mathfrak{X}(E))$.

The equations for vector fields are \ref{dfn:solvector}:
\begin{align}
    \inn_{Z}{\bar{\Omega}}=0\,,\quad \inn_{Z}\eta=0\,,\quad \inn_{Z}\kappa^i=0\,,\,i=1,\dots,n\,,\quad \inn_{Z}\tau=1\,.
\end{align}
 The last two equations are equivalent to saying that $Z$ is a holonomic vector field. The other two equations are equivalent to the cocontact Lagrangian equations \cite{LGGMR-2022}:
\begin{align}
    \inn_{Z}\d \hat{\eta}&=\d E_L-R_s(E_L)\eta-R_t(E_L)\tau\,,
    \\
    \quad \inn_{Z}\hat{\eta}&=-E_L\,.
\end{align}
Thus, the Herglotz variational problem leads to cocontact geometry. Notice that the same equations will have been found if we had considered other $R_s$ and $R_t$. The fact that $(\tau,\hat{\eta})$ is a cocontact system allows us to select a specially nice $R_s$ with the extra property that $ \inn_{R_s}\d\hat\eta=0$.

\section{Absorption of Constraints}\label{sec:abs}

A clearer presentation of the geometry ought to be provided by a variational problem without constraints, that is
\begin{align}
    (\Omega,\emptyset,\mathfrak{X}(E))\, \text{ or }(\Omega,\emptyset,\Gamma(V(\pi)))\,.
\end{align}
In this situation, all information is contained in the $2$-form $\Omega$. This means that the constraints will be recovered dynamically as new equations of motion. The archetypical example is the technique of the Lagrange multipliers (see, for instance, theorem 2 in \cite{gracia_geometric_2003}). Another relevant example are regular Lagrangian systems, where the holonomy conditions are also recovered dynamically (this is another nice property of the Poincaré-Cartan $2$-form \eqref{eq:PoincareCartan}). That is no the case if the Lagrangian is singular.

Here we present a general procedure to construct a new form such that the constraints are recovered dynamically at the cost of adding new variables. This procedure is heavily inspired on \cite{QC-2025}.

\subsection{Absorption of the Holonomy Condition of Lagrangian Systems}

We will use the same notation as in section \ref{sec:holLag}.
Consider the first jet bundle $\pi:E=J^1\rho\rightarrow \mathbb{R}$ of the fiber bundle $\rho:\mathbb{R}\times Q\rightarrow \mathbb{R}$, with $\dim Q=n$. The Lagrangian is a function $L:J^1\rho\rightarrow \mathbb{R}$. The holonomy condition is imposed by implementing the Cartan codistribution of $J^1\rho$ as vakonomic constraints. Therefore, $I_0=\emptyset$ and $I_1^{vak}=\langle\kappa^i\rangle_{i=1,\dots,n}$.

Concatenating the results of section \ref{sec:holLag}, we have that the VP $(L\tau,(\emptyset,I^{vak}_1),\textup{Adm})$ is dynamically equivalent to the GVP $(\d \Theta_L,(\emptyset,I^{vak}_1),\Gamma(V(\pi)))$\,. 

To absorb the constraints $I^{vak}_1$, we will add $n$ new variables. Consider the trivial bundle $\pi':J^1\rho\times\mathbb{R}^n\rightarrow J^1\rho$. The new coordinates will be denoted $p_i$ and they are globally defined. We will denote $L'=\pi'^*L$, $\kappa'^i=\pi'^*\kappa^i$ and, abusing notation, $\tau=\pi'^*\tau$.

\begin{thm}\label{thm:absholo1} The GVP $(\Omega_\mathfrak{U},\emptyset,\Gamma(V(\pi'\circ\pi)))$ is dynamically $\pi'$-equivalent to the GVP
\\$(\d \Theta_L,(\emptyset,I^{vak}_1),\Gamma(V(\pi)))$, with
    \begin{align}\label{eq:thetalvaknh2}
        \Omega_\mathfrak{U}=\d L'\wedge \tau+\d (p_i\kappa'^i)\,.
    \end{align}
\end{thm}

\begin{proof} Let $\gamma'$ be a solution of $(\Omega_\mathfrak{U},\emptyset,\Gamma(V(\pi'\circ\pi)))$. The local expression of an element $\xi'\in\Gamma(V(\pi'\circ\pi)))$ is
\begin{align*}
    \xi'=f^i\frac{\partial}{\partial q^i}+g^i\frac{\partial}{\partial v^i}+G_i\frac{\partial}{\partial p_i}\,, 
\end{align*}
for arbitrary $f^i, g^i,G_i\in C^\infty(J^1\rho\times\mathbb{R}^n)$. Then, the equations of motion of $(\Omega_\mathfrak{U},\emptyset,\Gamma(V(\pi'\circ\pi)))$ are:
\begin{align*}
\gamma'^*\inn_{\xi'}\Omega_\mathfrak{U}=\gamma'^*\left[
f^i\left(\frac{\partial L'}{\partial q^i}\tau-\d p_i\right)
+g^i\left(\frac{\partial L'}{\partial v^i}\tau- p_i\tau\right)
+G_i\left(\kappa'^i\right)
\right]=0\,.
\end{align*}
Therefore, we have the equations:
\begin{align}
\gamma'^*\kappa'^i=0\,,\label{eq:abshol1}
\\\label{eq:abshol2}
\gamma'^*\left(\frac{\partial L'}{\partial v^i}- p_i\right)=0\,,
\\\label{eq:abshol3}
\gamma'^*\left(\frac{\partial L'}{\partial q^i}\tau-\d \left(\frac{\partial L'}{\partial v^i}\right)\right)=0\,.
\end{align}

Define $\gamma=\pi'\circ\gamma'$, which is a section of $\pi$. The section $\gamma$ is holonomic (that is, it satisfies the constraints) because, as a consequence of \eqref{eq:abshol1}, $\gamma^*\kappa^i=\gamma'^*\kappa'^i=0$. For a generic vector field $\xi\in\Gamma(V(\pi))$ with local expression:
\begin{align*}
    \xi=f^i\frac{\partial}{\partial q^i}+g^i\frac{\partial}{\partial v^i}\,, 
\end{align*}
the equations of motion of the system $(\d \Theta_L,(\emptyset,I^{vak}_1),\Gamma(V(\pi)))$ are
\begin{align}\label{eq:abshol4}
 \gamma^*\inn_\xi\d\Theta=\gamma^*\left[
f^i\left(\frac{\partial L}{\partial q^i}\tau-\d \frac{\partial L}{\partial v^i}\right)
\right]\,,
\end{align}
which vanishes due to \eqref{eq:abshol3}. In conclusion, $\gamma=\pi'\circ\gamma'$ is a solution of $(\d \Theta_L,(\emptyset,I^{vak}_1),\Gamma(V(\pi)))$.

Consider we have two solutions $\gamma'_1$ and $\gamma'_2$ of $(\Omega_\mathfrak{U},\emptyset,\Gamma(V(\pi'\circ\pi)))$, which map to the same section $\pi'\circ \gamma'_1=\pi'\circ \gamma'_2$. They can only differ on the coordinates $p_i$. Nevertheless, as a consequence of \eqref{eq:abshol2},
\begin{align*}
    \gamma_1'^*p_i= \gamma_1'^*\frac{\partial L'}{\partial v^i}
    =(\pi'\circ\gamma_1')^*\frac{\partial L}{\partial v^i}
    =(\pi'\circ\gamma_2')^*\frac{\partial L}{\partial v^i}
    = \gamma_2'^*\frac{\partial L'}{\partial v^i}
    =\gamma_2'^*p_i\,.
\end{align*}
Therefore, $\gamma'_1=\gamma'_2$ and $\hat{\pi}'$ is injective on Sol$(\Omega_\mathfrak{U},\emptyset,\Gamma(V(\pi'\circ\pi)))$.

Finally, if $\gamma$ is a solution of $(\d \Theta_L,(\emptyset,I^{vak}_1),\Gamma(V(\pi)))$, define $\gamma'$ as the section of $\pi'\circ\pi$ satisfying $\pi'\circ\gamma'=\gamma$ with $\gamma'^*p_i=\gamma^*\frac{\partial L}{\partial v^i}$. Since $\gamma$ is holonomic, $\gamma'$ satisfies \eqref{eq:abshol1}. By the definition of the $p_i$, it automatically satisfies \eqref{eq:abshol2}. Finally, equation \eqref{eq:abshol3} is satisfied as long as \eqref{eq:abshol4} vanishes, which it does because $\gamma$ is a solution.
\end{proof}

If we identify $T^*Q\approx Q\times\mathbb{R}^n$, we can interpret $p_i$ as the momenta. The resulting formalism is the unified Lagrangian-Hamiltonian formalism, also known as Skinner-Rusk formalism \cite{SR-83}. The relation between the Skinner-Rusk formalism and the Cartan codistribution was first noticed by Gotay \cite{GotayCartan}. He also proved that both variational problems are actually dynamically equivalent. Nevertheless, as far as I know, Capriotti was the first to spell out the relation in similar terms as the ones exposed here \cite{capriotti}.

\subsection{Absorption of Nonholonomic Constraints}

Consider a system with constraints given by $(I_1^{nh})$. Assume that $I_1^{nh}$ is co-orientable, with $I_1^{nh}=\langle \eta^\alpha\rangle_{\alpha=1,\dots,k}$. 
Let $R_\alpha$ be vector fields satisfying \eqref{eq:reeb}, and $\{\frac{\partial}{\partial y^i}\}$ a local base of vertical vector fields which vanish when contracted by any $\eta^\alpha$ (that is, a base of Adm as given in \eqref{eq:admNonholo}). Given a $2$-form $\Omega\in\Omega^2(E)$, remember from section \ref{sec:nonhol} that $\bar{\Omega}$ is given by
\begin{align}
    \bar{\Omega}=\Omega+\inn_{R_\alpha}\Omega\wedge\eta^\alpha\,.
\end{align} 

Following the results and notation of section \ref{sec:nonhol}, the GVP $(\Omega,(\emptyset,I_1^{nh}),\textup{Adm})$ is dynamically equivalent to the GVP $(\bar{\Omega},(\emptyset,I_1^{nh}),\Gamma(V(\pi)))$.\footnote{In this section we will not consider constraints given by functions, as it is well-known that they can be incorporated using Lagrange multipliers. Avoiding them will make the exposition tidier.}

To absorb the constraints, we will add $k$ variables. Consider the trivial bundle $\pi':E\times\mathbb{R}^{k}\rightarrow E$. The new coordinates will be denoted by $p_\alpha$, and they are globally defined. We will denote ${\Omega}'=\pi'^*{\Omega}$, $\eta'^\alpha=\pi'^*\eta^\alpha$ and $\tau'=\pi'^*\tau$. Moreover, $R'_\alpha$ and $\frac{\partial}{\partial y^i}'$ are the lifts by the $0$-section of $R_\alpha$ and $\frac{\partial}{\partial y^i}$, respectively.

\begin{thm}\label{thm:absnoholo} The GVP $(\Omega_\mathfrak{U},\emptyset,\Gamma(V(\pi'\circ\pi)))$ is an extension of the GVP
\\$(\bar{\Omega},(\emptyset,I_1^{nh}),\Gamma(V(\pi)))$, with
    \begin{align}\label{eq:thetalvaknh2}
        \Omega_\mathfrak{U}=\Omega'  +\d p_\alpha\wedge\eta'^\alpha\,,
    \end{align}
\end{thm}

\begin{proof} Let $\gamma'$ be a solution of $(\Omega_\mathfrak{U},\emptyset,\Gamma(V(\pi'\circ\pi)))$. The local expression of an element $\xi'\in\Gamma(V(\pi'\circ\pi)))$ is
\begin{align*}
    \xi'=f^i\frac{\partial}{\partial y^i}'+F^\alpha R_\alpha'+G_\alpha\frac{\partial}{\partial p_\alpha}\,, 
\end{align*}
where $f^i, F^\alpha,G_\alpha\in C^\infty(E\times\mathbb{R}^k)$.

Then, the equations of motion of $(\Omega_\mathfrak{U},\emptyset,\Gamma(V(\pi'\circ\pi)))$ are:
\begin{align*}
\gamma'^*\inn_{\xi'}\Omega_\mathfrak{U}=\gamma'^*\left[
f^i\inn_{\frac{\partial}{\partial y^i}'}{\Omega}'
+F^\alpha\left(\inn_{R_\alpha'}\Omega'-\d p_\alpha\right)
+G_\alpha\eta'^\alpha
\right]=0\,.
\end{align*}
Therefore, we have the equations:
\begin{align}
\gamma'^*\eta'^\alpha=0\,,\label{eq:absnohol1}
\\\label{eq:absnohol2}
\gamma'^*\left(\inn_{R_\alpha'}\Omega'-\d p_\alpha\right)=0\,,
\\\label{eq:absnohol3}
\gamma'^*\left(\inn_{\frac{\partial}{\partial y^i}'}{\Omega}'\right)=0\,.
\end{align}

Define $\gamma=\pi'\circ\gamma'$, which is a section of $\pi$. The section $\gamma$ satisfies the constraints because, as a consequence of \eqref{eq:absnohol1}, $\gamma^*\eta^\alpha=\gamma'^*\eta'^\alpha=0$. For a generic vector field $\xi\in\Gamma(V(\pi))$ with local expression:
\begin{align*}
    \xi=f^i\frac{\partial}{\partial y^i}+F^\alpha R_\alpha\,, 
\end{align*}
the equations of motion of the system $(\bar{\Omega},(\emptyset,I_1^{nh}),\Gamma(V(\pi)))$ are
\begin{align}\label{eq:absnohol4}
 \gamma^*\inn_\xi\bar{\Omega}=\gamma^*\left(f^i\inn_{\frac{\partial}{\partial y^i}}{\Omega}\right)\,,
\end{align}
which vanishes due to \eqref{eq:absnohol3}. In conclusion, $\gamma=\pi\circ\gamma'$ is a solution of $(\bar{\Omega},(\emptyset,I_1^{nh}),\Gamma(V(\pi)))$.

Let $\gamma$ be a solution of $(\bar{\Omega},(\emptyset,I_1^{nh}),\Gamma(V(\pi)))$ with domain of definition $I_\gamma\subset\mathbb{R}$, which it is an open interval by definition. Define $\gamma'$ as the section of $\pi'\circ\pi$ satisfying $\pi'\circ\gamma'=\gamma$ and $\gamma'^*p_i$ defined as follows. The $1$-forms $\gamma^*\inn_{R_\alpha}\Omega$ are closed in $\mathbb{R}$ and, since $I_\gamma$ is contractible, there exist functions $G_\alpha$  such that $\d G_\alpha=\gamma^*\inn_{R_\alpha}\Omega$. Define $\gamma'^*p_\alpha(t)=G_\alpha(t)$.

Since $\gamma$ satisfies the constraints, $\gamma'$ satisfies \eqref{eq:absnohol1}. By the definition of the $p_\alpha$, it automatically satisfies \eqref{eq:absnohol2}. Finally, equation \eqref{eq:absnohol3} is satisfied as long as \eqref{eq:absnohol4} vanishes, which it does because $\gamma$ is a solution. Therefore, $\gamma'$ is a solution of $(\Omega_\mathfrak{U},\emptyset,\Gamma(V(\pi'\circ\pi)))$.

Consider we have two solutions $\gamma'_1$ and $\gamma'_2$ of $(\Omega_\mathfrak{U},\emptyset,\Gamma(V(\pi'\circ\pi)))$, which map to the same section $\pi'\circ \gamma'_1=\pi'\circ \gamma'_2$. They can only differ on the coordinates $p_\alpha$. As a consequence of \eqref{eq:absnohol2}, 
\begin{align*}
    \gamma_1'^*\d p_\alpha= \gamma_1'^*\inn_{R_\alpha'}\Omega'
    =(\pi'\circ\gamma_1')^*\inn_{R_\alpha'}\Omega'
    =(\pi'\circ\gamma_2')^*\inn_{R_\alpha'}\Omega'
    = \gamma_2'^*\inn_{R_\alpha'}\Omega'
    =\gamma_2'^*\d p_\alpha\,.
\end{align*}
Therefore, $\gamma'^*_1p_\alpha$ and $\gamma'^*_2p_\alpha$ may differ by a constant. In particular, the map $\hat{\pi}'$ is not injective on Sol$(\Omega_\mathfrak{U},\emptyset,\Gamma(V(\pi'\circ\pi)))$.
\end{proof}

\subsection{Absorption of the Holonomy Condition and Nonholonomic Constraints}

Consider a system with nonholonomic constraints given by $(I_1^{nh})$, and vakonomic cosntraints given by a Cartan codistribution $I^{vak}_1=\langle \kappa^i\rangle$.

To absorb the constraints, we will add $n+k$ variables. Consider the trivial bundle $\pi':J^1\rho\times N\times\mathbb{R}^{n+k}\rightarrow E=J^1\rho\times N$. The new coordinates will be denoted by $p_i$ and $p_\alpha$, and they are globally defined.

If the constraints are compatible (see section \ref{sec:nhvak}), then there exist $k$ vector fields $R_\alpha\in\mathfrak{X}(E)$ such that 
\begin{align}
\inn_{R_\alpha}\tau=0\,,\quad \inn_{R_\alpha}\eta^\beta=\delta^\beta_\alpha\,.
    \end{align}
    
We also have adapted coordinates $(t,q^i,v^i,s^\alpha, u^a)$ such that:
\begin{align}
    \tau=\d t\,,\quad \kappa^i=\d q^i-v^i\d t\,,\quad \eta^\alpha= \d s^\alpha+b_i^\alpha \d q^i+h^\alpha\d t +c^\alpha_a\d u^a\,,
\end{align}
for some functions $b^\alpha_ i,h^\alpha, c^\alpha_a\in C^\infty(E)$. Recall that, for a Lagrangian function $L:J^1\rho\times N\rightarrow\mathbb{R}$
\begin{align}
\bar{\Omega}=\d\left(L\d t+\frac{\partial L}{\partial v^i}\kappa^i\right)+\sigma_\alpha\wedge \eta^\alpha,
\end{align}
with $\sigma_\alpha=\inn_{R_\alpha}\d\left(L\d t+\frac{\partial L}{\partial v^i}\kappa^i\right)$.

\begin{thm}\label{thm:absnoholo} The GVP $(\Omega_\mathfrak{U},\emptyset,\Gamma(V(\pi'\circ\pi)))$ is an extension of the GVP
\\$(\bar{\Omega},(\emptyset,I_1^{nh},I_1^{vak}),\Gamma(V(\pi)))$, with
    \begin{align}\label{eq:thetalvaknh2}
        \Omega_\mathfrak{U}=\d L'\wedge\d t  +\d(p_ i\kappa'^i)+\d p_\alpha\wedge\eta'^\alpha\,.
    \end{align}
\end{thm}

\begin{proof} Let $\gamma'$ be a solution of $(\Omega_\mathfrak{U},\emptyset,\Gamma(V(\pi'\circ\pi)))$. The local expression of an element $\xi'\in\Gamma(V(\pi'\circ\pi)))$ is
\begin{align*}
    \xi'=f^i\frac{\partial}{\partial q^i}'
    + g^i \frac{\partial}{\partial v^i}'
    + F^\alpha \frac{\partial}{\partial s^\alpha}'
    + H^a\frac{\partial}{\partial u^a}'
    + G_i\frac{\partial}{\partial p_i}
    + G_\alpha\frac{\partial}{\partial p_\alpha}\,, 
\end{align*}
where $f^i, g^i,F^\alpha,H^a, G_i, G_\alpha\in C^\infty(J^1\rho\times N\times\mathbb{R}^{n+k})$. 

Then, the equations of motion of $(\Omega_\mathfrak{U},\emptyset,\Gamma(V(\pi'\circ\pi)))$ are:
\begin{align*}
\gamma'^*\inn_{\xi'}\Omega_\mathfrak{U}=&\gamma'^*\biggl[
f^i\left({\frac{\partial L'}{\partial q^i}}\d t - \d p_i-\inn_{\frac{\partial }{\partial q^i}'}\eta'^\alpha \d p_\alpha\right)
+g^\alpha\left({\frac{\partial L'}{\partial v^i}}\d t - p_i\d t\right)
\\
&
+F^\alpha\left(\frac{\partial L'}{\partial s^\alpha}\d t
-\d p_\alpha\right)
+H^a\left(\frac{\partial L'}{\partial u^a}-\inn_{\frac{\partial}{\partial u^a}'}\eta'^\alpha \d p_\alpha\right)
\\&
+G_i\kappa'^i
+G_\alpha\eta'^\alpha
\biggr]=0\,.
\end{align*}
Therefore, we have the equations:
\begin{align}
\gamma'^*\eta'^\alpha&=0\,,\label{eq:absvaknohol1}\,\quad  &\gamma'^*\kappa'^i=0\,,
\\
\gamma'^*\left(\frac{\partial L'}{\partial s^\alpha}\d t
-\d p_\alpha\right)&=0\,,\label{eq:absvaknohol4}\quad&
\gamma'^*\left({\frac{\partial L'}{\partial v^i}}\d t - p_i\d t\right)=0\,,
\\\label{eq:absvaknohol6}
\gamma'^*\left({\frac{\partial L'}{\partial q^i}}\d t - \d p_i-\inn_{\frac{\partial }{\partial q^i}'}\eta'^\alpha \d p_\alpha\right)&=0\,, &\gamma'^*\left(\frac{\partial L'}{\partial u^a}\d t-\inn_{\frac{\partial}{\partial u^a}'}\eta'^\alpha \d p_\alpha\right)=0\,.
\end{align}

Define $\gamma=\pi'\circ\gamma'$, which is a section of $\pi$. The section $\gamma$ satisfies the constraints as a consequence of \eqref{eq:absvaknohol1}. Due to \eqref{eq:absvaknohol4} we have that $ \gamma'^*p_i=\gamma^*\frac{\partial L}{\partial v^i}$ and $\gamma'^*\d p_\alpha=\gamma^*\frac{\partial L}{\partial s^\alpha}$. Moreover, $\gamma'^*\d p_\alpha=\gamma^*\sigma_\alpha$. Indeed, if we take $\xi'=R_\alpha'$:
\begin{align*}
0=\gamma'^*\inn_{R_\alpha'}\Omega_\mathfrak{U}=&\gamma'^*\inn_{R_\alpha'}\left(\pi'^*\d\Theta_L+\d\left(\left(p_i-\frac{\partial L'}{\partial v^i}\right)\kappa'^i\right)+\d p_\alpha\wedge\eta'^\alpha\right)=\gamma'^*\left(\pi'^*\inn_{R_\alpha}\d\Theta_L-\d p_\alpha\right)\,.
\end{align*}

For a generic vector field $\xi\in\Gamma(V(\pi))$ with local expression:
\begin{align*}
    \xi=f^i\frac{\partial}{\partial q^i}
    + g^i \frac{\partial}{\partial v^i}
    + F^\alpha \frac{\partial}{\partial s^\alpha}
    + H^a\frac{\partial}{\partial u^a}\,,  
\end{align*}
the equations of motion of the system $(\bar{\Omega},(\emptyset,I_1^{nh},I_1^{vak}),\Gamma(V(\pi)))$ are
\begin{align}\label{eq:absvaknohol7}
\gamma^*\inn_{\xi}\bar{\Omega}=\gamma^*\biggl[
f^i\left({\frac{\partial L}{\partial q^i}}\d t - \d \frac{\partial L}{\partial v^i}-\inn_{\frac{\partial}{\partial q^i}}\eta^\alpha\sigma_\alpha\right)
+F^\alpha\left(\frac{\partial L}{\partial s^\alpha}\d t-\sigma_\alpha\right)
+H^a\left(\frac{\partial L}{\partial u^a}-\inn_{\frac{\partial}{\partial u^a}}\eta^\alpha\sigma_\alpha\right)\biggr]=0\,.
\end{align}
Substituting \eqref{eq:absvaknohol4} into \eqref{eq:absvaknohol6} and using that $\gamma'^*\d p_\alpha=\gamma^*\sigma_\alpha$, we see that \eqref{eq:absvaknohol7} is satisfied. In conclusion, $\gamma=\pi\circ\gamma'$ is a solution of $(\bar{\Omega},(\emptyset,I_1^{nh},I_1^{vak}),\Gamma(V(\pi)))$.

Let $\gamma$ be a solution of $(\bar{\Omega},(\emptyset,I_1^{nh},I_1^{vak}),\Gamma(V(\pi)))$ with domain of definition $I_\gamma\subset\mathbb{R}$, which is an open interval by definition. Define $\gamma'$ as the section of $\pi'\circ\pi$ satisfying $\pi'\circ\gamma'=\gamma$, $\gamma'^*p_i=\gamma^*\frac{\partial L}{\partial v^i}$,
and $\gamma'^*p_\alpha$ defined as follows. The $1$-forms $\gamma^*\frac{\partial L}{\partial s^\alpha}\d t$ are closed in $\mathbb{R}$ and, since $I_\gamma$ is contractible, there exist functions $G_\alpha$  such that $\d G_\alpha=\gamma^*\frac{\partial L}{\partial s^\alpha}\d t$. Define $\gamma'^*p_\alpha(t)=G_\alpha(t)$.

Since $\gamma$ satisfies the constraints, $\gamma'$ satisfies \eqref{eq:absvaknohol1}. By the definition of the $p_\alpha$ and $p_i$, it automatically satisfies \eqref{eq:absvaknohol4}. From \eqref{eq:absvaknohol7} we derive the relation $\gamma'^*\d p_\alpha=\gamma^*\sigma_\alpha$. Finally, equations \eqref{eq:absvaknohol6} are satisfied as long as equations \eqref{eq:absvaknohol7} vanish, which they do because $\gamma$ is a solution. Therefore, $\gamma'$ is a solution of $(\Omega_\mathfrak{U},\emptyset,\Gamma(V(\pi'\circ\pi)))$.

Consider we have two solutions $\gamma'_1$ and $\gamma'_2$ of $(\Omega_\mathfrak{U},\emptyset,\Gamma(V(\pi'\circ\pi)))$, which map to the same section $\pi'\circ \gamma'_1=\pi'\circ \gamma'_2$. They can only differ on the coordinates $p_\alpha$. As a consequence of \eqref{eq:absvaknohol4}, 
\begin{align*}
    \gamma_1'^*\d p_\alpha
    =\gamma_2'^*\d p_\alpha\,.
\end{align*}
Therefore, $\gamma'^*_1p_\alpha$ and $\gamma'^*_2p_\alpha$ may differ by a constant. In particular, the map $\hat{\pi}'$ is not injective on Sol$(\Omega_\mathfrak{U},\emptyset,\Gamma(V(\pi'\circ\pi)))$.

\end{proof}

\section{Application: Singular Time-Dependent Lagrangians}\label{sec:singtime}

The cosymplectic approach to time-dependent Lagrangians takes place on $\pi:E=\mathbb{R}\times TQ\rightarrow \mathbb{R}$, which is a trivialization of the first jet bundle $J^1\rho$, of the fiber bundle $\rho:\mathbb{R}\times Q\rightarrow\mathbb{R}$ used in section \ref{sec:vak}. From a Lagrangian $L\in C^\infty(E)$ one can construct the Lagrangian $2$-form
\begin{align}
    \omega_L=-\d \frac{\partial L}{\partial v^i}\wedge\d q^i\,.
\end{align}

If $L$ is regular, the pair $(\omega_L,\tau)$ is a cosymplectic structure on $E$. Namely, $\tau$ and $\omega_L$ are closed and
\begin{align}
    \tau\wedge\omega_L^n\neq0\,.
\end{align}
Associated with any cosymplectic structure there is a unique vector field $R_t$ known as the Reeb vector field, such that~\cite{munoz-lecanda_geometry_2024}
\begin{align}
    \inn_{R_t}\tau=1\,,\quad \inn_{R_t}\omega_L=0\,.
\end{align}
The local expression of this vector field is \eqref{reeb:vak}. The corresponding field equations for a holonomic vector field $Z$ are:
\begin{align}\label{eq:eomcosym}
    \inn_Z\tau=1\,,\quad \inn_Z\omega_L=\d E_L- R_t(E_L)\tau\,.
\end{align}

When $L$ is not regular $(\omega_L,\tau)$ is a precosymplectic structure. That is, $\tau$ and $\omega_L$ are still closed, but 
\begin{align}
    \tau\wedge\omega_L^n=0\,.
\end{align}
In this case, if a Reeb vector field exists it is not unique. More significant, and sometimes overlooked in the literature, a Reeb vector field may not exist at all. For instance, if $n=1$ and $L=tv$, the Lagrangian $2$-form is $\omega_L=\d t\wedge\d q$. Clearly, $\ker \omega\subset \ker\tau$. In this situation equations \eqref{eq:eomcosym} cannot even be written down.

The discussion in sections \ref{sec:trans} and \ref{sec:vak} provides a solution for these cases. The variational principle for the Lagrangian $L=tv$ exists, and one can write down the Euler-Lagrange equations. The problem is that the precosymplectic structure $(\omega_L,\tau)$ is not the correct geometry to describe the dynamics. In section \ref{sec:trans} we have seen that it is enough to take any vector field $R_t$ such that $\inn_{R_t}\tau=1$. Then, the corresponding $\omega$ is derived from $L$ and also from $R_t$ (see  \eqref{eq:formstrans}). We can always take 
\begin{align}
    R_t=\frac{\partial}{\partial t}\,,
\end{align}
regardless of the regularity of the Lagrangian. Then,
\begin{align}
   \sigma_t&=-R_t(E_L)\tau+\d E_L+\frac{\partial^2 L}{\partial t\partial v^i}\d q^i\,,
   \\
  \omega &=\left[\d \left(\frac{\partial L}{\partial v^i}\right)-\frac{\partial^2 L}{\partial t\partial v^i} \d t\right]\wedge \d q^i=-\omega_L-\frac{\partial^2 L}{\partial t\partial v^i} \d t\wedge \d q^i.
\end{align}
By construction, $\inn_{R_t}\tau=1$ and $\inn_{R_t}\omega=0$. Therefore, $(\omega,\tau)$ is a precosymplectic structure with Reeb vector field for all Lagrangians. The field equations for a holonomic vector field $Z$ are:
\begin{align}
    \inn_Z\tau=1\,,\quad \inn_Z\omega=\d E_L- R_t(E_L)\tau+\Lie_{R_t}\left(\frac{\partial L}{\partial v^i}\d q^i\right)\,.
\end{align}
The existence of a solution of these equations is equivalent to the existence of a vertical vector field $X=Z-R_t$ such that the hypotheses of theorem \ref{theo:trans} hold. When no such solution exists, a constraint algorithm should be applied. The quantity $\frac{\partial L}{\partial v^i}\d q^i$ can be defined intrinsically as $\inn_J\d L$, where $J$ is the canonical endomorphism \cite{munoz-lecanda_geometry_2024}.

In conclusion, the dynamics of a time-dependent Lagrangian can always be recovered by a (pre)cosymplectic structure, but the $2$-form $\omega$ cannot always be the Lagrangian $2$-form $\omega_L$. The alternative pair $(\omega,\tau)$ with
\begin{align}
\omega=-\omega_L-\frac{\partial^2 L}{\partial t\partial v^i} \d t\wedge \d q^i=-\inn_J\d L+\left(\inn_{R_t}(\d\inn_J\d L)
\right)\wedge\tau
\end{align}
always has a Reeb vector field (is always precosymplectic) and provides the desired dynamics.

\section{Application: Singular Action-Dependent Lagrangians}\label{sec:singaction}

The Herglotz variational principle for regular Lagrangians is associated with contact geometry \cite{HerglotzGuenther}, and its nonautonomous counterpart to cocontact geometry \cite{LGGMR-2022}. Precontact and precocontact geometry describe some singular Lagrangians but not all \cite{LL-2019,LGGMR-2022}. In this section we will focus in those singular Lagrangians for which the standard procedure does not lead to pre(co)contact geometry (non-admissible Lagrangians, in the notation proposed in \cite{LGGMR-2022}). We will derive their dynamical equations geometrically and we will identify its geometry.

There are several equivalent definitions of a precontact structure. Here we present one characterization given in \cite{grabowska}. A \textbf{precontact} structure on a $(2n+1)$-dimensional manifold $\mathfrak{E}$ is given by the distribution $\ker(\bar\eta)$ of a $1$-form $\bar\eta$ such that $\d\bar\eta$ has rank $2r$ on $\ker(\bar\eta)$, with $0\leq r< n$. A precontact structure has a Reeb vector field, that is, $R\in\mathfrak{X}(\mathfrak{E})$ such that
\begin{align}\label{eq:reebcontact}
    \inn_R\bar\eta=1\,, \quad \inn_R\d\bar\eta=0\,,
\end{align}
although it is not unique, unlike in the contact case. Nevertheless, all Reeb vector fields lead to the same equations of motion \cite{LL-2019}, which are, for a vector field $Z\in\mathfrak{X}(\mathfrak{E})$ and a Hamiltonian $H\in C^\infty(\mathfrak{E})$:
\begin{align}\label{eq:eomcontact}
    \inn_Z\d \bar\eta=\d H-R(H)\bar\eta\,,\quad\inn_Z\bar\eta=-H\,.
\end{align}

Lagrangian (pre)contact systems on a manifold $\mathfrak{E}=TQ\times\mathbb{R}$ are given by a Lagrangian function $L:\mathfrak{E}\rightarrow \mathbb{R}$. The Hamiltonian is the Lagrangian energy $E_L$, with the same expression as \eqref{eq:lagrangianenergy}. Moreover, from the Lagrangian one can construct the 1-form~\cite{GGMRR-2019b}:
\begin{align}
    \eta_L=\d s-\frac{\partial L}{\partial v^i}\d q^i\,.
\end{align}
If the Lagrangian is regular, then $\eta_L$ defines a contact structure. If not, it may lead to a precontact structure, but not always. The prototypical Lagrangian that does not lead to a precontact structure is 
\begin{align}\label{eq:lagsingcontact}
L=vs\,,    
\end{align}
where $(q,v,s)$ are local coordinates of $\mathfrak{E}=T\mathbb{R}\times\mathbb{R}$. The corresponding 1-form is $\eta_L=\d s-s\d q$, and the kernels are the distributions generated at each point by
\begin{align}
    \ker \eta_L=\left\langle\frac{\partial}{\partial v},s\frac{\partial}{\partial s}+\frac{\partial}{\partial q}\right\rangle\,,\quad  \ker \d\eta_L=\left\langle\frac{\partial}{\partial v}\right\rangle\,.
\end{align}
The rank of $\ker \d\eta_L$ on $\ker\eta_L$ is $1$ and, thus, $\ker(\eta)$ is not a precontact structure. In particular, there does not exist a vector field satisfying \eqref{eq:reebcontact} for $\eta_L$. Therefore, it is not possible to write down the precontact equations \eqref{eq:eomcontact}. The situation for precocontact systems, that is, non-autonomous action-dependent Lagrangians, is similar.

\subsection{Geometric Derivation of the Herglotz-Euler-Lagrange Equations}

Consider $E=J^1\rho\times\mathbb{R}$, with local coordinates $(t,q^i,v^i,s)$. A non-autonomous action-dependent Lagrangian is a function $L:E\rightarrow \mathbb{R}$. As explained in section \ref{sec:nhvak}, the Herglotz variational problem is $(L\tau,(\emptyset,I^{nh}_1=\langle\eta\rangle,I_1^{vak}=\langle\kappa^i\rangle_{i=1,\dots,n}),\textup{Adm})$, where $\eta$ is
\begin{align}\label{eq:etanhvak2}
    \eta=\d s+\left(\frac{\partial L}{\partial v^i}v^i-L\right) \d t-\frac{\partial L}{\partial v^i}\d q^i=\eta_L+E_L\d t\,,
\end{align} 
 and $\kappa^i$ is the Cartan codistribution of $J^1\rho$ pullbacked to $E$. Due to theorems \ref{thm:base},  \ref{thm:vaknh1} and \ref{thm:vaknh2} it is dynamically equivalent to the GVP $    (\bar{\Omega},(\emptyset,I^{nh}_1,I_1^{vak}),\Gamma(V(\pi)))$, with
\begin{align}\label{eq:newthing}
    \bar{\Omega}=\d\Theta_L+\sigma_s\wedge\eta\,,\quad \Theta_L=L\d t+\frac{\partial L}{\partial v^i}\kappa^i\,,\quad \sigma_s=\inn_{R_s}\d\Theta_L\,.
\end{align}
Following the derivation of section \ref{sec:nhvak}, we see that it is not required that $\eta$ in \eqref{eq:etanhvak2} or $\eta_L$ are precontact. We only need that two vector fields exist, $R_t$ and $R_s$, such that
\begin{align}
    \inn_{R_t}\tau=1\,,\quad \inn_{R_s}\tau=0\,,\quad \inn_{R_s}\eta=1\,.
\end{align}

Indeed, consider a generic $R_s$ and $R_t$. In local coordinates $(t,q^i,v^i,s)$
\begin{align*}
    R_s=A^i\frac{\partial}{\partial q^i}+B^i\frac{\partial}{\partial v^i}+\left(1+A^i\frac{\partial L}{\partial v^i}\right)\frac{\partial}{\partial s}\,,
\end{align*}
with $A^i,B^i\in C^{\infty}(E)$. Then
\begin{align*}
    \sigma_s=\frac{\partial L}{\partial s}\d t+A^i\left(\frac{\partial L}{\partial q^i}\d t-\d\frac{\partial L}{\partial v^i}+\frac{\partial L}{\partial s}\frac{\partial L}{\partial v^i}\d t\right)+R_s\left(\frac{\partial L}{\partial v^i}\right)\kappa^i\,.
\end{align*}

A section $\gamma$ is solution of $(\bar{\Omega},(\emptyset,I^{nh}_1,I_1^{vak}),\Gamma(V(\pi)))$ if:
\begin{align}\label{eq:contactsing}
  \gamma^*\inn_\xi\bar{\Omega}=0\,, \forall{\xi}\in\Gamma(V(\pi))\,,\quad  \gamma^*\eta=0\,,\quad \gamma^*\kappa^i=0\,.  
\end{align}
A generic vector field $\xi\in\Gamma(V(\pi))$ has local expression:
\begin{align*}
    \xi=f^i\frac{\partial}{\partial q^i}
    + g^i \frac{\partial}{\partial v^i}
    + F\frac{\partial}{\partial s}\,,  
\end{align*}
with $f^i, g^i,F\in C^\infty(E)$. For $j=1,\dots,n$ define $\xi_j$ as the vector field $\xi_j\in\Gamma(V(\pi))$ with $f^j=1$, $f^i=0$ if $i\neq j$, $g^i=0$ and $F=\frac{\partial L}{\partial v^j}$. Then, if $\gamma$ satisfies the constraints,
\begin{align}\label{eq:singLagEq}
     0=\gamma^*\inn_{\xi_j}\bar{\Omega}=\gamma^*\left(\frac{\partial L}{\partial q^j}\d t-\d\frac{\partial L}{\partial v^j}+\frac{\partial L}{\partial s}\frac{\partial L}{\partial v^j}\d t\right)\,.
\end{align}
In particular, $\gamma^*\sigma_s=\gamma^*\frac{\partial L}{\partial s}\d t$. Finally, a lengthy but straightforward computations shows that, if $\gamma$ is a section that satisfies the constraints and \eqref{eq:singLagEq} then $\gamma^*\inn_\xi\bar{\Omega}=0$ for all vertical vector fields $\xi$.

Equations \eqref{eq:singLagEq}, together with the constraint $\gamma^*\eta=0$, are the Herglotz-Euler-Lagrange equations, derived geometrically for an arbitrary Lagrangian. Notice that the resulting equations do not depend on the $R_s$ or $R_t$ chosen, as expected from theorem \ref{thm:vaknh2}. This equations do not correspond to those of a precocontact or precosymplectic systems. In fact,
\begin{align}
    \d \bar\Omega=\d \sigma_s\wedge \eta-\sigma_s\d\eta\,.
\end{align}
If $\d\sigma_s=0$, corresponding to Lagrangians with a closed action dependence (see \cite{GLMR-2022}), and using the fact that $\d\eta=-\d\Theta_L$,
\begin{align}
    \d \bar\Omega=\sigma_s\wedge\bar\Omega\,.
\end{align}
Therefore, $\bar\Omega$ is locally conformally symplectic (l.c.s) (see \cite{CARINENA2025105418} for a recent presentation) if it is a regular $2$-form (for instance, if the Lagrangian is regular).  Nevertheless, some Lagrangians, including \eqref{eq:lagsingcontact}, do not have the property $\d \sigma_s=0$.

\subsection{Absorption of constraints}
 Consider the trivial bundle $\pi':J^1\rho\times \mathbb{R}\times\mathbb{R}^{n+1}\rightarrow J^1\rho\times \mathbb{R}$. The local coordinates are $(t,q^i,v^i,s,p_i,p_s)$. 
 From theorem \ref{thm:absnoholo}, the GVP $(\Omega_\mathfrak{U},\emptyset,\Gamma(V(\pi'\circ\pi)))$ is an extension of the GVP
\\$(\bar{\Omega},(\emptyset,I_1^{nh},I_1^{vak}),\Gamma(V(\pi)))$, with
    \begin{align}\label{eq:thetalvaknh2}
        \Omega_\mathfrak{U}=\d L'\wedge\d t  +\d(p_ i\kappa'^i)+\d p_s\wedge\eta'\,.
    \end{align}

The apostrophe denotes objects lifted to $J^1\rho\times \mathbb{R}\times\mathbb{R}^{n+1}$ via pull-back or zero-section. The dynamics of the system are completely encoded in the $2$-form $\Omega_\mathfrak{U}$. It is not closed in general:
\begin{align*}
    \d\Omega_\mathfrak{U}=-\d p_s\wedge \d \eta'=\d p_s\wedge\Omega_{\mathfrak{u}}-\d p_s\wedge \d \left(\left(p_i-\frac{\partial L}{\partial v^i}\right)\kappa'^i\right)\,.
\end{align*}

The last term disrupts the identification of $\Omega_\mathfrak{U}$ as l.c.s. However, this term is identically $0$ on solutions. This motivates the following reformulation. Define
\begin{align}
   \check\eta=\d s-p_i\kappa'^i- L'\d t,\,\quad \check\Omega=\d L'\wedge\d t  +\d(p_ i\kappa'^i)+\d p_s\wedge\check\eta=-\d \check\eta +\d p_s\wedge\check\eta\,.
\end{align}
\begin{thm}
    The GVP $(\Omega_\mathfrak{U},\emptyset,\Gamma(V(\pi'\circ\pi)))$ is dynamically equivalent to the GVP $(\check\Omega,\emptyset,\Gamma(V(\pi'\circ\pi)))$
\end{thm}
\begin{proof}
    Let $\gamma$ be a section of $\pi'\circ\pi$ and $\xi\in\Gamma(V(\pi'\circ\pi))$ with the local expression
\begin{align*}
    \xi=f^i\frac{\partial}{\partial q^i}
    + g^i \frac{\partial}{\partial v^i}
    + F \frac{\partial}{\partial s}
    + G_i\frac{\partial}{\partial p_i}
    + G\frac{\partial}{\partial p_s}\,, 
\end{align*}
where $f^i, g^i,F, G_i, G\in C^\infty(J^1\rho\times \mathbb{R}\times\mathbb{R}^{n+1})$. 
Then, the equations of motion of $(\Omega_\mathfrak{U},\emptyset,\Gamma(V(\pi'\circ\pi)))$ are:
\begin{align*}
\gamma^*\inn_{\xi}\Omega_\mathfrak{U}=&\gamma'^*\biggl[
f^i\left({\frac{\partial L'}{\partial q^i}}\d t - \d p_i+\frac{\partial L'}{\partial v^i}\d p_s\right)
+g^i\left({\frac{\partial L'}{\partial v^i}}\d t - p_i\d t\right)
\\
&
+F\left(\frac{\partial L'}{\partial s}\d t
-\d p_s\right)+G_i\kappa'^i
+G\eta'
\biggr]=0\,.
\end{align*}
On the other hand, the equations of motion of
$(\check\Omega,\emptyset,\Gamma(V(\pi'\circ\pi)))$ are:
\begin{align*}\label{eq:34}
\gamma^*\inn_{\xi}\check\Omega=&\gamma'^*\biggl[
f^i\left({\frac{\partial L'}{\partial q^i}}\d t - \d p_i+p_i\d p_s\right)
+g^i\left({\frac{\partial L'}{\partial v^i}}\d t - p_i\d t\right)
\\
&
+F\left(\frac{\partial L'}{\partial s}\d t
-\d p_s\right)+G_i\kappa'^i
+G\check\eta
\biggr]=0\,.
\end{align*}
Both systems of equations impose that $\gamma^*\frac{\partial L'}{\partial v^i}=\gamma^*p_i$, which makes the systems of equations equal. Therefore, they have the same solutions.
\end{proof}

The $2$-form $\check\Omega$ satisfies that $\d\check\Omega=\d p\wedge\check\Omega$. I would be l.c.s as defined in \cite{CARINENA2025105418} if $\check\Omega$ where regular. A simple computation in coordinates shows that it is regular when $p_i\neq\frac{\partial L'}{\partial v^i}$ and singular otherwise. Another difference is that it is a non-autonomous formulation, where the Lagrangian is already embedded in $\check\Omega$. This changes how the dynamics are defined\footnote{The Hamiltonian vector field $X$ as defined in \cite{CARINENA2025105418} with the same dynamics as the GVP $(\check\Omega,\emptyset,\Gamma(V(\pi'\circ\pi)))$ is the one related to the function $F=0$ and such that $\inn_X\d t=1$.}. Having said that, we will analyse $\check\Omega$ with the apprach presented in \cite{CARINENA2025105418} for l.c.s.

The $2$-form $\check\Omega$ is of ``first kind'' because there exists a vector field $U=\frac{\partial}{\partial p_s}$, called the Lee vector field, with the properties\footnote{In \cite{CARINENA2025105418} they denote $-\d p$ as $\eta$.}:
\begin{align}
 \Lie_U\check\Omega=0\,, \quad\inn_U\d p=1\,.   
\end{align}

In particular, it is $\d_{-\d p_S}$-exact:
\begin{align*}
    \d_{-\d p_s}(-\check\eta)=-\d\check\eta+\d p_s\wedge\check\eta=\check\Omega\,.
\end{align*}

Due to it being singular, there does not exists a characteristic vector field $\Gamma_{\check\Omega}$ with the property that $\inn_{\Gamma_{\check\Omega}}\check\Omega=-\d p$. Therefore, $\check\Omega$ conforms to a non-autonomous singular version of l.c.s of first kind which do not induce a bicontact structure $(\d p_s,\check\eta)$ (see \cite{bicontact,CARINENA2025105418}).

We can make explicit the failure of $\check{\eta}$ being (pre)contact, steaming form $(\d p_s,\check\eta)$ not being a bicontact structure due to $\check\Omega$ not being regular. Lets consider a generic $R$
\begin{align*}
R= f\frac{\partial}{\partial t}+ f^i\frac{\partial}{\partial q^i}
    + g^i \frac{\partial}{\partial v^i}
    + F \frac{\partial}{\partial s}
    + G_i\frac{\partial}{\partial p_i}
    + G\frac{\partial}{\partial p_s}\,. 
\end{align*}

If $\inn_R\check\eta=1$, then $F=1+p_if^i-f(p_iv^i-L')$. Then
\begin{align*}
\inn_R\d\check\eta&=-\left(f^i{\frac{\partial L'}{\partial q^i}}+g^i\left(\frac{\partial L'}{\partial v^i}-p_i\right)+F\frac{\partial L'}{\partial s}-G_iv^i\right)\d t-\left(G_i-f\frac{\partial L'}{\partial q^i}\right)\d q^i
\\&
-f\left(p_i-\frac{\partial L'}{\partial v^i}\right)\d v^i
+f\frac{\partial L'}{\partial s}\d s
+\left(f^i-fv^i\right)\d p_i=0\,.
\end{align*}
The problematic condition is the one corresponding to $\d t$. Using the other conditions and $\inn_R\check\eta=1$:
\begin{align*}
 fv^i{\frac{\partial L'}{\partial q^i}}+g^i\left(\frac{\partial L'}{\partial v^i}-p_i\right)+\frac{\partial L'}{\partial s}+\frac{\partial L'}{\partial s}fL'-f{\frac{\partial L'}{\partial q^i}}v^i=g^i\left(\frac{\partial L'}{\partial v^i}-p_i\right)+\frac{\partial L'}{\partial s}=0\,,   
\end{align*}
which has no solution if $\frac{\partial L'}{\partial v^i}-p_i=0$ and $\frac{\partial L'}{\partial s}\neq0$.
Notwithstanding, the approach to the Herglotz variational principal presented in this paper is more similar to multicontact \footnote{Multicontact geometry was intriduced in \cite{LGMRR-2022b} and generalized in \cite{bracketsmultico}. A survey and comparison of related structures can be found in \cite{SurveyAdFT}. For a more hands-on presentation, see \cite{deLeon:2024ztn}.}. In multicontact geometry, apart of asking that $\inn_R\check\eta=1$ and $\inn_R\d t=0$, it is only ask that $\inn_R\d\check\eta$ is $(\pi'\circ\pi)$-semibasic. The vector fields that satisfy this conditions are
\begin{align}\label{eq:reebsfinal}
    R= \frac{\partial}{\partial s}+g^i \frac{\partial}{\partial v^i}
    + G\frac{\partial}{\partial p_s}\,,
\end{align}
for arbitrary $g^i, G\in C^\infty(J^1\rho\times\mathbb{R}\times\mathbb{R}^{n+1})$. Namely, $\check\eta$ is not precontact in general, but it is premulticontact.

So far the discussion has been valid for any action-dependent Lagrangian $L$. What sets apart the non-admissible Lagrangians like \eqref{eq:lagsingcontact} from the admissible ones? One of the equations of motion \eqref{eq:34} is $p_i-\frac{\partial L'}{\partial v^i}=0$. This has to be understood as a constraint that defines a submanifold $\iota:\mathcal{S}\hookrightarrow  J^1\rho\times\mathbb{R}\times\mathbb{R}^{n+1}$ where solutions can exists. The differential forms $\check\Omega$ and $\check\eta$ can be transported to $\mathcal{S}$ using the smooth embedding $\iota$. Nonetheless, $\iota^*\check\eta$ may not be premulticontact on $\mathcal{S}$. That is the case for the Lagrangian \eqref{eq:lagsingcontact}.

If any of the possible Reeb vector field \eqref{eq:reebsfinal} is tangent to the constraint surface, then $\iota^*\check\eta$ is premulticontact. Locally, such a vector exists if there exists functions $g^i\in C^\infty(J^1\rho\times\mathbb{R}\times\mathbb{R}^{n+1})$ such that
\begin{align}\label{eq:reebsfinal2}
    g^i\frac{\partial^2L}{\partial v^i\partial v^j}+\frac{\partial^2L}{\partial s\partial v^j}=0\,.
\end{align}

This is not a necessary condition. For instance, consider $n=2$ with coordinates $(t,q^a,q^b,v^a,v^b,s,p_a,p_b,p_s)$ and the Lagrangian $L=sv^a+v^b$. Equation \eqref{eq:reebsfinal2} is never satisfied but 
\begin{align*}
    \iota^*\check\eta=\d s-s\d q^a-\d q^b\,
\end{align*}
has a Reeb vector field, $R=-\frac{\partial}{\partial q^b}$.

The constraints absorbed include the holonomy constraints. It is reasonable that the resulting structures are similar to the Lagrangian-Hamiltonian unified approach developed for contact mechanics in \cite{deLeon:2020wdu}. From that formalism a similar discussion about the singular action-dependent Lagrangians can be made. 

\section{Summary and Outlook}

\hspace{\parindent} We have presented a constructive method to find a geometric structure that recover the solutions of a variational problem. To do so, we have introduced the concept of geometric variational problem, a rather general idea which encompasses different geometrical structures describing dynamical systems on the same footing. This definition facilitates the comparison between them and other variational problems. We have shown how to identify a particular geometric structure from a variational problem in concrete situations characterized by how the constraints select the admissible variations. This includes the holonomy condition understood as a vakonomik constraint, and nonholonomic constraints, as a generalization of Herglotz's variational principle. In particular, we recovered the cases of symplectic, cosymplectic, contact and cocontact geometry, for Hamiltonian and regular Lagrangian systems. Therefore, the method agrees with previous results when applied to well-known situations.

We have applied the method to singular non-autonomous Lagrangian systems. We have shown that the resulting geometry is always (pre)cosymplectic, although it is not always possible to use the Lagrangian $2$-form. We have proposed a modification of the $2$-form that leads to a (pre)cosymplectic structure for any Lagrangian.

We have also applied the method to singular action-dependent Lagrangians, and shown that it does not always lead to precontact geometry. The dynamics can be recovered from a more general geometry \eqref{eq:newthing}, similar to contact geometry but without demanding the existence of a Reeb vector field.

We also present a method to incorporate the constraints as equations of motion. When applied to Hamilton's variational principle it leads to the Skinner-Rusk unified formalism. For action-dependent Lagrangians, the method leads to a non-autonomous and non-regular version of locally conformally symplectic geometry. Non admissible Lagrangians, like \eqref{eq:lagsingcontact}, are explained as those that, even though they start as premulticontact structures, they are not premulticontact on the constraint surface.

\paragraph{Outlook.} The method presented here is a first step to solving the problem of identifying the correct geometry of a dynamical system. To arrive at a satisfactory answer significant advances are still needed.

First of all, the method does not provide uniqueness. On the contrary, there are steps where a vector field has to be chosen, leading to different geometric objects. Nevertheless, all of them lead to the same dynamics and are different instances of the same kind of geometric structure. This is closer to the notion of equivalent Lagrangians, or Lepage equivalent forms.  In this regard, the techniques presented here can be adapted to detect equivalent geometric structures. Notwithstanding, these results do not preclude the existence of qualitatively different geometric structures with the same dynamics. The question of what are the intrinsic geometric objects that univocally characterize a dynamical system remains open.

Secondly, the results hinge on the use of globally defined $1$-forms and vector fields. The study of contact and $k$-contact geometries shows that it is more insightful to work with distributions \cite{GGKM-2024,deLucas:2024kei}. An approach using distributions could extend the result to more general constraints. Of special interest are $k$-contact distributions \cite{deLucas:2024kei}, which include the Cartan distribution and are related to nonholonomic constraints. A better understanding of these distributions, including weaker versions, will enable the generalization of the results presented here.

The method of absorbing the constraints leads to an important observation. All different geometric structures considered lead to the same kind of geometric structure: the dynamics is given by the kernel of a $2$-form. The subtleties of the different kinds of geometric structures is encoded in the particular evolution of some variables. When adding the variable, the interpretation appears clear, as we already expect that the dynamical evolution of the new variable to be special. But, for consistency, we should be aware of the contrary: any variable of a system whose evolution is special in the same way, should be considered an indication that the system of study is, in fact, an extension of a more fundamental one. Examples of this particular dynamics are constraints and auxiliary variables (they do not appear in the equations). The identification of these variables is closely related to the counting degrees of freedom used in physics (The method is credited to Dirac. For a more geometrical approach see \cite{Gracia:1988xp}, and for a recent development and discussion see \cite{ErrastiDiez:2023gme}).

A similar observation can be made regarding the processes of symplectification, contactification and reconstruction. Of particular interest are results like \cite{GGKM-2024,BLMP-2020}, where the presence of a particular distribution or submanifold leads to the identification of a particular kind of dynamics inside a bigger system. 

A takeaway message from the results presented here is that the constraints can change the underlying geometry of the system. In this regard it will be interesting to understand the relation between the results presented here and nonholonomic mechanics.
This is an alternative way of implementing nonholonomic constraints in dynamical systems, which gives less emphasis to the variations (see \cite{dLdD-1996}, and \cite{LLLM-2023,JL-2025} for recent developments). 

Finally, the application of the method to Herglotz's variational principle leads to a new geometric structure \eqref{eq:newthing}, similar to contact geometry but for a generic $1$-form $\eta$. This structure is a reduction of a non-autonomous non-regular version of l.c.s. geometry. Structures adjacent to l.c.s appear recurrently when studying contact geometry \cite{bracketsmultico,GGKM-2024}. To better characterize the dynamics arising from non-holonomic constraints (in particular, contact mechanics), it will be beneficial to develop l.c.s geometry for non-autonomous and singular systems.

\section*{Acknowledgements}
This paper is dedicated to the memory of M.C. Muñoz-Lecanda. He introduced me to the concept of nonholonomic constraints and noticed their importance in contact mechanics. 
This is the key insight that vertebrates this work. 
His inquisitive questions prompted me to think more deeply about geometric mechanics. 
The ensuing critical pondering has been instrumental in setting the problem in the terms presented in this work.

The author wishes to express his gratitude to the referees, whose valuable feedback and observations have significantly improved the manuscript.

The author acknowledges financial support of the Ministerio de Ciencia, Innovación y Universidades (Spain), PID2021-125515NB-21.


\begin{thebibliography}{99}


\bibitem{AM-78}
R.~Abraham, and J.E.~Marsden,
{\it Foundations of Mechanics} ($2nd$ ed.),
Benjamin--Cummings, Redwood City CA, 1987.
(ISBN-10: 080530102X).

\bibitem{Atiyah}
M.F. Atiyah,
``Convexity and commuting Hamiltonians'',
{\sl Bull. London Math.
Soc.} \textbf{14} (1), (1982).
(\url{ https://doi.org/10.1112/blms/14.1.1})

\bibitem{Bell:2025kxx}
 C. Bell, and D. Sloan
``Dynamical Similarity in Multisymplectic Field Theory,''
(\url{https://arXiv:2509.16099})

\bibitem{bicontact}
D.E. Blair, G.D. Ludden, and K. Yano,
``Geometry of complex manifolds similar to the Calabi-Eckmann manifolds'',
{\sl J. Differ. Geom.} \textbf{9}, (1974), 263–274.

\bibitem{Bravetti2017}
A.~Bravetti,
\newblock {``Contact Hamiltonian dynamics: the concept and its use''},
\newblock {\sl Entropy} {\bf 19}(10) (2017) 535.
(\url{https://doi.org/10.3390/e19100535})

\bibitem{BLMP-2020}
A. Bravetti, M. de Le\'on, J.C. Marrero, and E. Padr\'on,
``Invariant measures for contact Hamiltonian systems: symplectic sandwiches with contact bread'',
{\sl J. Phys. A: Math. Gen.} {\bf 53}(45) (2020) 455205.
(\url{https://doi.org/10.1088/1751-8121/abbaaa})

\bibitem{capriotti}
 S. Capriotti, 
``Unified formalism for Palatini gravity'',
{\sl  Int. J. Geom. Meth. Mod. Phys.} {\bf 15}(3) (2018) 1850044.
(\url{https://doi.org/10.1142/S0219887818500445})

\bibitem{CARINENA2025105418}
J.F. Cariñena, and P. Guha
``Lichnerowicz-Witten differential, symmetries and locally conformal symplectic structures''
{\sl J. Geom. Phys.} {\bf 210} (2025) 105418.
(\url{https://doi.org/10.1016/j.geomphys.2025.105418})

\bibitem{C1981} M. Crampin, “On the differential geometry of the Euler--Lagrange equations, and
the inverse problem of Lagrangian dynamics”, 
{\sl J. Phys. A: Math. Gen. } {\bf 14}(1981)  2567–
2575.
(\url{https://doi.org/10.1088/0305-4470/14/10/012})

\bibitem{CLLL}
L. Colombo, M. de León, M. Lainz and A. López-Gordón,
``Liouville-Arnold theorem for contact Hamiltonian systems''
(\url{https://arxiv.org/abs/2302.12061})


\bibitem{dLdD-1996}
M. de León, and D.M. de Diego
``On the geometry of non‐holonomic Lagrangian systems'',
{\sl  J. Math. Phys.} {\bf 37}(7) (1996) 3389--3414.
(\url{https://doi.org/10.1063/1.531571})


\bibitem{LeSaVi2016}
M.~de~Le\'on, M.~Salgado, and S.~Vilari\~no,
\emph{Methods of Differential Geometry in Classical Field Theories: $k$-Symplectic and $k$-Cosymplectic Approaches},
World Scientific, Hackensack, 2016.
(\url{http://doi.org/10.1142/9693})

\bibitem{deLeon:2016vms}
M.~de Le\'on, and C.~Sard\'on,
``Cosymplectic and contact structures for time-dependent and dissipative Hamiltonian systems,''
J. Phys. A \textbf{50} (2017) no.25, 255205.
(\url{https://doi.org/10.1088/1751-8121/aa711d})

\bibitem{LL-2018}
M.~{de Le{\'o}n}, and M.~{Lainz--Valc{\'a}zar},
``Contact Hamiltonian systems'',
{\sl J. Math. Phys.} {\bf 60}(10) (2019) 102902.
(\url{https://doi.org/10.1063/1.5096475})

\bibitem{LL-2019}
M. de León, and M. Lainz,
``Singular Lagrangians and precontact Hamiltonian Systems"
{\sl Int. J. Geom. Methods Mod. Phys.} {\bf 16} (2019).
(\url{https://doi.org/	10.1142/S0219887819501585})

\bibitem{deLeon:2020wdu}
M.~de Le{\'o}n, J.~Gaset, M.~Lainz, X.~Rivas, and N.~Rom{\'a}n-Roy,
``Unified Lagrangian-Hamiltonian Formalism for Contact Systems,''
Fortsch. Phys. \textbf{68} (2020) no.8, 2000045
(\url{https://oi:10.1002/prop.202000045})

\bibitem{LLM-2021}
M. de León, M. Lainz, and M. C. Muñoz-Lecanda,
``The Herglotz Principle and
Vakonomic Dynamics'', In F. Nielsen and F. Barbaresco, editors
{\sl Geometric Science of Information} {\bf 12829} of {\sl Lecture Notes in Computer Science} (2021) 183-190.
(\url{https://doi.org/10.1007/978-3-030-80209-7_21})

\bibitem{LGL-2021}
M. de Le\'on, J. Gaset, and M. La\'inz, 
``Inverse problem and equivalent contact systems'', 
{\sl J. Geom. Phys.} {\bf  176} (2022) 104500. 
(\url{https://doi.org/10.1016/j.geomphys.2022.104500})

\bibitem{LGGMR-2022}
M.~de~Le{\'{o}}n, J. Gaset, X. Gr\`acia, M.C. Mu\~noz-Lecanda, and X. Rivas,
``Time-dependent contact mechanics'',
{\sl Monatsh. Math.} {\bf 201}:1149--1183 (2023).
(\url{https://doi.org/10.1007/s00605-022-01767-1})

\bibitem{LGMRR-2022b}
M.~de~Le{\'{o}}n, J. Gaset, M.C. Mu\~noz-Lecanda, X. Rivas, and N. Rom\'an-Roy,
``Multicontact formulation for non-conservative field theories'',
\newblock {\sl J. Phys. A: Math. Theor.} {\bf 56}(2) (2023) 025201.
(\url{https://doi.org/10.1088/1751-8121/acb575})

\bibitem{LLM-2023}
M. de León, M. Lainz, and M.C. Muñoz-Lecanda,
``Optimal control, contact dynamics
and Herglotz variational problem'',
{\sl J. Nonlinear Sci.} {\bf 33}(1) (2023) 9.
(\url{https://doi.org/10.1007/s00332-022-09861-2})

\bibitem{LLLM-2023}
M. de León, M. Lainz, A. López-Gordón, and J.C. Marrero,
``A new perspective on nonholonomic brackets and Hamilton–Jacobi theory'',
{\sl J. Geom. Phys.}\textbf{ 198}, (2024) 105116
(\url{https://doi.org/10.1016/j.geomphys.2024.105116})


\bibitem{deLeon:2024ztn}
M.~de Le\'on, J.~Gaset Rif\`a, M.C.~Mu\~noz-Lecanda, X.~Rivas, and N.~Rom\'an-Roy,
``Practical Introduction to Action-Dependent Field Theories,''
{\sl Fortsch. Phys.} (2025).
(\url{https://doi.org/10.1002/prop.70000})

\bibitem{bracketsmultico}
M.~de León, R. Izquierdo-López, and X.~Rivas,
``Brackets in multicontact geometry and multisymplectization,''
\url{https://arXiv.2505.13224}

\bibitem{deLucas:2024kei}
J.~de Lucas, X.~Rivas, and T.~Sobczak,
``Foundations on k-contact geometry,''
(\url{https://arxiv.org/abs/2409.11001})

\bibitem{DH-2009}
E. Dyer, and K. Hinterbichler,
``Boundary Terms, Variational Principles and Higher Derivative Modified Gravity,''
Phys. Rev. D \textbf{79} (2009), 024028
(\url{doi:10.1103/PhysRevD.79.024028})


\bibitem{echeverria-enriquez_geometry_1996}
A. Echeverría-Enríquez, M.C. Muñoz-Lecanda, and N. Román-Roy,
``Geometry of Lagrangian first-order classical field theories",
{\sl Fortsch. Phys.} {\bf 44}, 235--280, (1996).
(\url{https://doi.org/10.1002/prop.2190440304})

\bibitem{ErrastiDiez:2023gme}
V.~Errasti D\'\i{}ez, M.~Maier, and J.A.~M\'endez-Zavaleta,
``Constraint characterization and degree of freedom counting in Lagrangian field theory,''
{\sl Phys. Rev. D} \textbf{109} (2024) no.2, 2
(\url{doi:10.1103/PhysRevD.109.025010})

\bibitem{PLG-1974}
P.L. García
``The Poincaré-Cartan invariant in the calculus of variations'', 
{\sl Symp. Math.} {\bf XIV} (1974) 219-246

\bibitem{GGMRR-2019b}
J.~{Gaset}, X.~{Gr\`acia}, M.C.~{Mu\~noz--Lecanda}, X.~{Rivas}, and
  N.~{Rom\'an--Roy},
``New contributions to the Hamiltonian and Lagrangian contact formalisms for dissipative mechanical systems and their symmetries'',
{\sl Int. J. Geom. Meth. Mod. Phys.} {\bf 17}(6) (2020) 2050090.
(\url{https://doi.org/10.1142/S0219887820500905})


\bibitem{GGMRR-2019}
J.~{Gaset}, X.~{Gr\`acia}, M.C.~{Mu\~noz-Lecanda}, X.~{Rivas}, and
  N.~{Rom\'an-Roy},
``A contact geometry framework for field theories with dissipation'',
\newblock {\sl Ann. Phys.} {\bf 414} (2020) 168092.
(\url{https://doi.org/10.1016/j.aop.2020.168092})

\bibitem{GGMRR-2020}
J.~{Gaset}, X.~{Gr\`acia}, M.C.~{Mu\~noz-Lecanda}, X.~{Rivas}, and N.~{Rom\'an-Roy},
``A $k$-contact Lagrangian formulation for nonconservative field theories'',
{\sl Rep. Math. Phys.} {\bf 87}(3) (2021) 347--368.
(\url{https://doi.org/10.1016/S0034-4877(21)00041-0})

\bibitem{GRL-2023}
J. Gaset, X. Rivas, and A. López-Gordón,
``Symmetries, conservation and dissipation in time-dependent contact systems'', 
{\sl Fortsch. Phys.} {\bf 71}(8-9) (2023) 2300048.
(\url{https://doi.org/10.1002/prop.202300048})

\bibitem{GLMR-2022}
J. Gaset, M. La\'inz, A. Mas, and X. Rivas,
``The Herglotz variational principle for dissipative field theories'', 
{\sl Geom. Mech.} {\bf 1}(2) (2024) 153--178.
(\url{https://doi.org/10.1142/S2972458924500060})

\bibitem{SurveyAdFT}
J. Gaset Rifà, X.~Rivas, and N. Román-Roy
``A survey on geometric frameworks for action-dependent classical field theories and their relationship''
(\url{https://doi.org/10.48550/arXiv.2505.13224})

\bibitem{GBVP}
F. Gay-Balmaz
``Clebsch Variational Principles in Field Theories and Singular Solutions of Covariant Epdiff Equations''
{\sl Rep. Math. Phys} {\bf 1}(71), 2 (2013) 231--277.
(\url{https://doi.org/10.1016/S0034-4877(13)60032-4})

\bibitem{GotayCartan}
 M.J. Gotay, 
``An exterior differential system approach to the Cartan form, symplectic geometry and mathematical physics''. 
{\sl Actes du Coll\`oque de G\'eom\'etrie Symplectique et Physique Math\'ematique en l'honneur de Jean-Marie Souriau}, (Aix-en-Provence, France,1990), 1991, 160--188.

\bibitem{grabowska}
K. Grabowska, and J. Grabowski,
``Reductions: precontact versus presymplectic",
{\sl 	Ann. Mat. Pura Appl.} {\bf 202}, 2803--2839 (2023).
(\url{https://doi.org/10.1007/s10231-023-01341-y})

\bibitem{GGKM-2024}
K. Grabowska, J. Grabowski, M. Kuś, and G. Marmo, 
``Contactifications: a Lagrangian description of compact Hamiltonian systems'', 
{\sl J. Phys. A: Math.Theor.} {\bf  57} (2024) 395204.
(\url{https://doi.org/10.1088/1751-8121/ad75d8}).

\bibitem{Gracia:1988xp}
X.~Gràcia, and J.M.~Pons,
``Gauge Generators, Dirac's Conjecture and Degrees of Freedom for Constrained Systems,''
{\sl Annals Phys.} \textbf{187} (1988), 355.
(\url{https://doi.org/10.1016/0003-4916(88)90153-4})


\bibitem{gracia_geometric_2003}
X. Gràcia, J. Marín-Solano, and M.C. Muñoz-Lecanda,
``Some geometric aspects of variational calculus in constrained systems",
{\sl Rep. Math. Phys.} {\bf 51}, 127--148 (2003).
(\url{https://doi.org/10.1016/S0034-4877(03)80006-X})

\bibitem{Griffiths}
P.A. Griffiths,
``Exterior Differential Systems and the Calculus of Variations'', 
Birkhäuser Boston, MA. (1983). (\url{https://doi.org/10.1007/978-1-4615-8166-6})

\bibitem{HerglotzGuenther}
R.B. Guenther, C.M. Guenther, and J.A. Gottsch,
``The Herglotz Lectures on contact Transformations and Hamiltonian System'',
Lecture Notes in Nonlinear analysis,
Juliusz Schauder Center for Nonlinear Analysis, Nicholas Copernicus University.

\bibitem{GMPS-2015}
V. Guillemin, E. Miranda, A.R. Pires, and G. Scott
``Toric Actions on b-Symplectic Manifolds",
{\sl Int. Math. Res. Not.} {\bf 2015}(14), (2025) 5818–5848.
(\url{https://doi.org/10.1093/imrn/rnu108})

\bibitem{JL-2025}
V.M. Jimenez, and M. de León,
``The nonholonomic bracket on contact mechanical systems",
{\sl J. Geom. Phys.} {\bf 213}, (2025) 105484.
(\url{https://doi.org/10.1016/j.geomphys.2025.105484})

\bibitem{Krupka}
D. Krupka,
''Introduction to global variational geometry", Amsterdam: Atlantis Press, (2015). (\url{doi:10.2991/978-94-6239-073-7})

\bibitem{Lichnerowicz}
A. Lichnerowicz,
“Les variétés de Jacobi et leurs algèbres de Lie associées,”
{\sl J. Math. Pures Appl.}. Neuvième Série \textbf{57} (1978),  453–488.

\bibitem{Margalef-Bentabol:2020teu}
J.~Margalef-Bentabol, and E.J.S.~Villase\~nor,
``Geometric formulation of the Covariant Phase Space methods with boundaries,''
{\sl Phys. Rev. D} \textbf{103} (2021) no.2, 025011
(\url{https://doi.org/10.1103/PhysRevD.103.025011})


\bibitem{MW}
J. Marsden, A. Weinstein,
``Reduction of symplectic manifolds with symmetry'', {\sl Rep. Math. Phys.}, {\bf 5}(1) (1974) 121--130. 
(\url{https://doi.org/10.1016/0034-4877(74)90021-4})

\bibitem{MO-2023}
E. Miranda, C. Oms,
``Contact structures with singularities: From local to global'', {\sl J. Geom. Phys.}, {\bf 192} (2023) 104957. 
(\url{https://doi.org/10.1016/j.geomphys.2023.104957})

\bibitem{munoz-lecanda_geometry_2024}
M.C.~Muñoz-Lecanda, and N. Román-Roy,
``Geometry of Mechanics",
(\url{https://arxiv.org/abs/2401.12650})

\bibitem{QC-2025}
G. Quijón, and S. Capriotti
``A Hamiltonian formalism for general variational problems, with applications to first order gravity with basis",  	(\url{https://arxiv.org/abs/2407.03991})


\bibitem{Roman-Roy:2005vwe}
N.~Roman-Roy,
``Multisymplectic Lagrangian and Hamiltonian formalism of first-order classical field theories,''
{\sl SIGMA}\textbf{ 5} (2005), 100.
(\url{https://doi.org/10.3842/SIGMA.2009.100})


\bibitem{SR-83}
R. Skinner, and R. Rusk,
``Generalized Hamiltonian dynamics I: Formulation on $T^*Q\otimes TQ$'', 
{\sl J. Math. Phys.} {\bf 24}(11) 2589-2594 (1983) .
(\url{https://doi.org/10.1063/1.525654})

\bibitem{Sloan:2024kzb}
D.~Sloan,
``Dynamical similarity in field theories'',
{\sl Class. Quant. Grav.}, 
 \textbf{42}, (2025) 045001.
(\url{https://doi.org/ 10.1088/1361-6382/ada714})

\bibitem{TTW-2000}
B. Tsang, S.W. Taylor, and G.C. Wake,
``Variational methods for boundary value problems''
{\sl J. Appl. Math. Decis. Sci.} {\bf 4}(2), 193--204 (2000).
 (\url{http://eudml.org/doc/121735})

\end{thebibliography}

\let\emph\oldemph

\bibliographystyle{abbrv}

\end{document}